\documentclass[useAMS,referee,usenatbib]{biom}

\usepackage{booktabs}     
\usepackage{multirow}     
\usepackage{array}        
\usepackage[ruled,vlined]{algorithm2e}
\usepackage{lineno}       
\usepackage{setspace}     
\usepackage{amsmath,amssymb}
\usepackage{graphicx}

\DeclareMathOperator*{\mean}{mean}

\DeclareMathOperator*{\argmax}{argmax}

\renewcommand{\S}{\mathbb{S}}

\title{Sparse Bayesian Partially Identified Models for Sequence Count Data}

\author{Won Gu$^{1}$, 
Francesca Chiaromonte$^{1,5}$, and Justin D. Silverman$^{1,2,3,4*}$\email{JustinSilverman@psu.edu} \\
$^{1}$Department of Statistics, Pennsylvania State University, University Park, PA, U.S.A. \\
$^{2}$College of Information Sciences and Technology, Pennsylvania State University, University Park, PA, U.S.A. \\
$^{3}$Department of Medicine, Pennsylvania State University, Hershey, PA, U.S.A. \\
$^{4}$Institute for Computational and Data Science, Pennsylvania State University, University Park, PA, U.S.A. \\
$^{5}$L’EMbeDS, Sant’Anna School of Advanced Studies, Pisa, Italy}

\begin{document}


\date{{\it Received October} 2025. {\it Revised XX} 20XX.  {\it Accepted XX} 20XX.}



\pagerange{\pageref{firstpage}--\pageref{lastpage}} 
\volume{xx}
\pubyear{2025}
\artmonth{XX}


\doi{10.1111/j.1541-0420.2005.00454.x}


\label{firstpage}


\begin{abstract}
In genomics, differential abundance and expression  analyses are complicated by the compositional nature of sequence count data, which reflect only relative--not absolute--abundances or expression levels. Many existing methods attempt to address this limitation through data normalizations, but we have shown that such approaches imply strong, often biologically implausible assumptions about total microbial load or total gene expression. Even modest violations of these assumptions can inflate Type I and Type II error rates to over 70\%. Sparse estimators have been proposed as an alternative, leveraging the assumption that only a small subset of taxa (or genes) change between conditions. However, we show that current sparse methods suffer from similar pathologies because they treat sparsity assumptions as fixed and ignore the uncertainty inherent in these assumptions. We introduce a sparse Bayesian Partially Identified Model (PIM) that addresses this limitation by explicitly modeling uncertainty in sparsity assumptions. Our method extends the Scale-Reliant Inference (SRI) framework to the sparse setting, providing a principled approach to differential analysis under scale uncertainty. We establish theoretical consistency of the proposed estimator and, through extensive simulations and real data analyses, demonstrate substantial reductions in both Type I and Type II errors compared to existing methods.
\end{abstract}

%



\maketitle


%

\section{Introduction}

Sequence count technologies such as
16S rRNA microbiome surveys and RNA-seq gene expression profiling are crucial tools in many scientific fields. A primary objective of these assays is to identify microbes or genes that differ in amount (abundance or expression) between biological conditions (e.g., health and disease). This task, known as 
differential abundance or differential expression analysis, is statistically challenging. The data consist of high-dimensional count vectors with many zeros or near-zeros, and the total counts are arbitrary—unrelated to the \textit{scale} of the biological systems being measured (e.g., total microbial load in the human gut or total mRNA content in a cell). As a result, such data are referred to as compositional, meaning they encode only relative—not absolute—quantities~\citep{gloor2017notoptional}. The difficulty of analyzing compositional data has been recognized for over 125 years~\citep{pearson1897spurious} and remains an open problem. For example, a microbial taxon may appear more abundant in fecal samples from a diseased cohort simply because other taxa have decreased their loads (a relative increase). Similarly, a gene may appear less expressed in samples excised from tumors because other genes have increased their activation (a relative decrease). 

The most common approaches 
for addressing this scale ambiguity are referred to as normalizations.
These include techniques that transform the data to enforce a shared scale. For example, Total Sum Scaling (TSS) divides each sample by its total count making each transformed datum sum to one~\citep{weiss2017normal}. Other methods impose parameter constraints to achieve identifiability, such as restricting parameters to a simplex~\citep{love2014moderated,srinivasan2020knockoff}. However, recent work has shown that these strategies implicitly encode strong, often biologically implausible, assumptions about scale~\citep{nixon2024a,nixon2024b}. For instance, TSS effectively assumes that all individuals in a microbiome study have the same total microbial load~\citep{mcgovern2024intervals}, an assumption rarely supported by experimental evidence~\citep{vandeputte2017quantitative}. Even modest deviations from such assumptions can 
severely undermine the accuracy of an analysis, including false discovery rates exceeding 70\%, reduced power, and Type I error rates that paradoxically increase with larger sample sizes~\citep{nixon2024a,nixon2024b,konnaris2025mlscale}.

The emerging field of Scale-Reliant Inference (SRI) offers an alternative~\citep{nixon2024b}. SRI recognizes that sequence count data lack direct information about absolute scale, yet scientifically  meaningful
estimands often depend on it. This results in non-identifiability: the data alone are insufficient to uniquely determine the parameters of interest. Rather than forcing identifiability through untestable assumptions, SRI uses Bayesian Partially Identified Models (PIMs), often implemented as \textit{Scale Simulation Random Variables} (SSRVs). SSRVs introduce a prior over the missing scale parameters, explicitly modeling and propagating uncertainty in scale assumptions. By embracing this uncertainty, SRI yields more transparent and robust inference, often dramatically reducing Type I error while preserving or enhancing power~\citep{gloor2023sri,bennett2025scale,dossantos2024aldexvaginal,nixon2024a,nixon2024b}.

Despite these advances, recent methods based on sparse modeling assumptions have been largely overlooked in the SRI framework.
These methods are distinct from normalization-based strategies and typically assume that only a small subset of microbial taxa or genes change in abundance or expression across conditions~\citep{lin2014variable,srinivasan2020knockoff,grantham2020mimix,susin2020variableselectionmicrobiome,scott2023bayesdamicrobiome,zhang2021bayessparseda,monti2024knockoffmicrobiome,lin2020ancombc}. Such assumptions can be violated~\citep{mcgovern2024intervals}, but are often justifiable~\citep{nixon2024b}. Existing sparse estimators differ substantially in how they encode sparsity, estimate scale, and interpret effect sizes. It is therefore unclear which, if any, provide reliable and biologically meaningful results.

This article makes several key contributions. First, we extend SRI to sparse modeling and identify a common but underappreciated pitfall in methods that apply penalized estimators with a sum-to-zero constraint. Rather than capturing true biological sparsity, these constraints impose unrealistic assumptions about changes in total abundance or expression, which can lead to biased and misleading inferences. Second, we propose a new perspective on sparsity, defining it in terms of the mode of the log-fold-change distribution, rather than enforcing it through penalized estimation. This definition of sparsity is flexible; 
it can be employed even when many microbial taxa or genes--possibly 
more than 50\%--are subject to change,
and it motivates a new class of mode-estimating scale models within the SSRV framework, resulting in Sparse
SSRVs. Third, we establish consistency results for Sparse 
SSRVs. 
Fourth, we benchmark these methods extensively on simulated and real datasets, showing that they frequently achieve lower Type I error rates than existing methods while maintaining, and often improving, power. Notably, they outperform previous (dense) SSRVs when sparsity assumptions are valid. Finally, unlike many sparse estimators, our models fail gracefully under misspecification, tending toward conservative inference that prioritizes Type I error control. Collectively, these results suggest that Sparse 
SSRVs provide a promising, interpretable, and reliable tool for differential analysis under scale uncertainty.

The remainder of this article is organized as follows. Section~\ref{sec:sri-review} reviews SRI, formulating differential analyses as a rank-1 update problem. Section~\ref{sec:pitfall-normalized-sparsity} critically examines existing methods, highlighting the pitfalls of sparsity assumptions with sum-to-zero constraints. Section~\ref{sec:sparse-bpim} formalizes our mode-estimating scale models within the Bayesian PIM framework. Section~\ref{sec:consistency} establishes theoretical consistency results. Section~\ref{sec:benchmarks-simulation} presents simulation-based benchmarks independent of our modeling assumptions, and Section~\ref{sec:benchmarks-real} benchmarks the methods on real datasets with direct measurements of absolute scale. Section~\ref{sec:discussion} discusses limitations and future directions.

\section{An Overview of SRI and SSRVs}
\label{sec:sri-review}
Consider an observed \(D \times N\) sequence count table \(Y\). Throughout, we focus on differential abundance in microbiome studies as our running example, where each element \(Y_{dn}\) denotes the number of reads mapping to taxon \(d\) in sample \(n\). In SRI, \(Y\) is viewed as a noisy, incomplete measurement of an underlying biological system \(W\). Like \(Y\), \(W\) is a \(D \times N\) matrix, but its elements \(W_{dn}\) represent the true absolute abundances of taxa in their respective systems (e.g., microbial communities). \(W\) can be uniquely decomposed into composition and scale:
\setlength{\abovedisplayskip}{5pt plus 1pt minus 1pt}
\setlength{\belowdisplayskip}{5pt plus 1pt minus 1pt}
\begin{align}
W_{dn} &= W^{\parallel}_{dn}\, W^{\perp}_n, \label{eq:wdecomp}\\
W^{\perp}_{n} &= \sum_{d=1}^D W_{dn}, \nonumber
\end{align}
where \(W^{\parallel}\) is a \(D \times N\) matrix of proportional abundances (columns lie on the simplex) and \(W^{\perp}\) is an \(N\)-vector of positive total abundances (scales).

SRI imposes minimal assumptions on \(W\) or \(Y\), apart from a basic \emph{identification restriction}: the distribution of \(Y\) identifies \(W^{\parallel}\) but not \(W^{\perp}\)~\citep{nixon2024b}. Even with infinite data, \(Y\) uniquely determines the relative composition \(W^{\parallel}\) but provides no information about the absolute scale \(W^{\perp}\)~\citep{nixon2024b,gloor2017notoptional}. Consequently, inference must address two uncertainty sources: (1) sampling variability in estimating \(W^{\parallel}\) from finite data and (2) fundamental, unresolved uncertainty about \(W^{\perp}\).

Typically, inference targets functions of \(W\) rather than its individual elements. For differential abundance analysis, a natural estimand is the log-fold change (LFC):
\setlength{\abovedisplayskip}{5pt plus 1pt minus 1pt}
\setlength{\belowdisplayskip}{5pt plus 1pt minus 1pt}
\begin{equation}\label{eq:lfc}
\theta_{d} = \mean_{n:x_{n}=1} (\log W_{dn}) - \mean_{n:x_{n}=0} (\log W_{dn}),
\end{equation}
for a condition indicator \(x_{n}\). Substituting the decomposition in Equation~\eqref{eq:wdecomp} gives:
\setlength{\abovedisplayskip}{5pt plus 1pt minus 1pt}
\setlength{\belowdisplayskip}{5pt plus 1pt minus 1pt}
\begin{align}
\theta_d &= \underbrace{\mean_{n:x_n=1}(\log W^{\parallel}_{dn}) - \mean_{n:x_n=0}(\log W^{\parallel}_{dn})}_{\theta^{\parallel}_{d}}
+ \underbrace{\mean_{n:x_n=1}(\log W^{\perp}_{n}) - \mean_{n:x_n=0}(\log W^{\perp}_{n})}_{\theta^{\perp}}. \label{eq:lfc-decomp}
\end{align}
Hence, the vector of true LFCs \(\theta=(\theta_1,\dots,\theta_D)^\top\) admits a rank-1 update:
\setlength{\abovedisplayskip}{5pt plus 1pt minus 1pt}
\setlength{\belowdisplayskip}{5pt plus 1pt minus 1pt}
\begin{equation}\label{eq:rank-1-update}
\theta = \theta^{\parallel} + \theta^{\perp}\mathbf{1}_{D},
\end{equation}
where \(\theta^{\parallel}\) captures compositional LFCs and \(\theta^{\perp}\) the shift due to scale. This highlights a key simplification: inference on \(\theta\) depends only on the scalar \(\theta^{\perp}\), not the full vector \(W^{\perp}\). However, without additional information about \(\theta^{\perp}\), \(\theta\) remains only partially identified.

As detailed in \ref{sec:general-rank-1}, this rank-1 update formulation extends beyond LFCs to any linear functional of log-abundances (\(W\)), including regression coefficients and penalized smoothers such as spline estimators. Although we focus on LFCs, the same framework applies to a broad class of linear analyses beyond differential abundance.

\subsection{Normalizations and Scale Simulation Random Variables}\label{sec:norm-bayes-part}

Classical methods achieve \emph{artificial identifiability} by imposing strong, often biologically implausible, assumptions about scale. Because sequence count data contain no information about the absolute scale \(W^{\perp}\), these approaches enforce identifiability by fixing \(\theta^{\perp}\) to a single value—typically by (1) transforming data to share a common total (e.g., Total Sum Scaling (TSS)~\citep{weiss2017normal}) or (2) constraining latent parameters (e.g., sum-to-zero or simplex restrictions~\citep{love2014moderated,srinivasan2020knockoff}). Such procedures implicitly assume identical total microbial loads across samples~\citep{mcgovern2024intervals,vandeputte2017quantitative}, an assumption rarely supported empirically. Even mild violations can inflate false discovery rates above 75\% and distort inference~\citep{nixon2024a,nixon2024b}. In short, artificial identifiability trades biological plausibility for statistical convenience--providing precise answers to questions the data cannot, in truth, resolve.

\emph{Scale-Reliant Inference (SRI)} takes the opposite view: rather than enforcing identifiability, it explicitly models the uncertainty in the unidentifiable scale parameter \(\theta^{\perp}\). This is achieved through \emph{Scale Simulation Random Variables} (SSRVs), which represent the joint model
\[
p(\theta^{\parallel},\theta^{\perp}\mid Y)
  = p(\theta^{\parallel}\mid Y)\,p(\theta^{\perp}\mid\theta^{\parallel})\ ,
  \qquad 
  \theta^{\perp}\perp Y \mid \theta^{\parallel}
\]
where \(p(\theta^{\parallel}\mid Y)\)  is the \textit{measurement model} and \(p(\theta^{\perp}\mid\theta^{\parallel})\) is the \textit{scale model}. 
This factorization highlights both the philosophical and computational advantages of SRI. Unlike classical models, which falsely imply that \(Y\) identifies \(W^{\perp}\) (and therefore \(\theta^{\perp}\)), through the conditional independence relationship, SSRVs explicitly encode that no amount of data can update beliefs about \(\theta^{\perp}\) beyond prior information.

Even simple stochastic scale models can substantially improve inference. For instance, replacing a fixed normalization (e.g.,
the centered log-ratio transformation used in ALDEx2~\citep{fernandes2014unifying}) with a Gaussian prior on \(\log W^{\perp}\) propagates uncertainty in total load, reducing false discoveries while preserving power~\citep{nixon2024a,nixon2024b}. In practice, SSRVs \emph{fail gracefully}--yielding conservative inference when scale information is weak, and gaining power only when scale assumptions are well supported.
Details on 
implementing SSRVs and an 
ALDEx2 example are provided in \ref{sec:aldex-example}.

\section{Reviewing Sparsity Assumptions and Pitfalls}
\label{sec:pitfall-normalized-sparsity}

Although SRI theory for differential analysis has advanced substantially~\citep{mcgovern2024intervals,mcgovern2023gsea,nixon2024a,nixon2024b}, existing approaches have not been optimized for settings where only a small subset of taxa change between conditions—that is, where \(\theta=(\theta_{1},\dots,\theta_{D})\) is sparse. Many current methods attempt to exploit sparsity to gain identifiability, but often do so by introducing strong or biologically implausible assumptions. 
Here, we highlight some representative pitfalls and partial solutions.

\paragraph{Ignoring uncertainty in \(W^{\parallel}\)} A major limitation of many sparse estimators is that they treat the relative composition \(W^{\parallel}\) as known given the observed counts \(Y\). Common practice replaces \(W^{\parallel}\) by empirical proportions \(\tilde{W}^{\parallel}_{dn}=Y_{dn}/\sum_{d}Y_{dn}\)~\citep{zhou2022linda,lin2014variable,srinivasan2020knockoff,susin2020variableselectionmicrobiome,mandal2015ancom}, or by deterministic normalization factors~\citep{love2014moderated,robinson2010edger,ritchie2015limma}. This neglects sampling uncertainty due to finite sequencing depth--a well-documented source of inflated
Type I error in count data analyses~\citep{chen2024multiaddgps,gloor2016uncertainty}.

\paragraph{Imposing artificial sparsity through parameter constraints}
Many sparse methods also enforce constraints intended to guarantee scale invariance.
For example, \citet{lin2014variable} estimate regression coefficients
\(\beta\) via
\setlength{\abovedisplayskip}{5pt}
\setlength{\belowdisplayskip}{5pt}
\begin{equation}\label{eq:penalized-estimator}
\hat{\beta}
=\underset{\beta}{\text{argmin}}
\left(\frac{1}{2N}\|X-Z'\beta\|_2^2+\lambda\|\beta\|_1\right)
\quad\text{s.t.}\quad\sum_{d=1}^{D}\beta_d=0,
\end{equation}
where \(Z=\log(\tilde{W}^{\parallel})=\log(Y_{dn}/\sum_d Y_{dn})\).
The sum-to-zero constraint, once thought necessary for identifiability~\citep{aitchison1982codapaper},
is now recognized as unnecessary~\citep{nixon2024b} and biologically implausible.
In real systems--such as communities exposed to narrow-spectrum antibiotics--total abundance can decreases,
so \(\sum_d \beta_d < 0\); the constraint in \eqref{eq:penalized-estimator}
forces an unrealistic symmetry between losses and gains across taxa.

To assess its impact, we performed simple simulations varying sparsity level
and the sign distribution of true effects. Results (\ref{sec:ckf-sim})
show that even mild violations of the sum-to-zero assumption cause severe inflation
of false discovery rates and loss of power, demonstrating that performance depends on how a method enforces its sparsity assumptions.

Such a limitation is not unique to sparse methods based on Equation~\eqref{eq:penalized-estimator}—each sparse method carries its own implicit expectation about the degree of sparsity, determined by the specific mechanism it uses to enforce it. Existing methods encode the notion that ``most features remain unchanged'' by modeling the parameter distribution as concentrated around zero. This is often done with a point mass model where g, the density of $\theta$, is spiked at zero (e.g., spike-and-slab), or with a continuous unimodal model whose maximum occurs at zero (e.g., Laplace or double-exponential for LASSO, horseshoe). However, even if the true density g has a unique maximum at zero, it may not be as concentrated as these priors imply;
the way they enforce concentration around zero carries an implicit expectation about the overall degree of
sparsity, which may not match the true sparsity level. As a result, performance can deteriorate, as shown in
\ref{sec:ckf-sim}. We will later demonstrate in Section 5 how different sparse methods behave under varying degrees of sparsity.


\paragraph{Partial solutions}
A few approaches, such as MIMIX~\citep{grantham2020mimix}, use multinomial likelihoods with Laplace priors and avoid explicit parameter constraints. However, their strong mean-zero priors effectively reintroduce a weak sum-to-zero assumption, and relaxing these priors often leads to computational instability. Other models, including ANCOM-BC2~\citep{lin2020ancombc} and LinDA~\citep{zhou2022linda}, estimate log-fold changes from CLR-normalized parameters and then correct for the bias introduced by the normalization. LinDA, in particular, assumes that if few taxa are differentially abundant, the bias from CLR normalization can be mitigated by centering the mode of the estimated LFC distribution at zero. As discussed below, this represents a more principled way to encode sparsity than the penalized estimator in \eqref{eq:penalized-estimator}; however, it does not propagate uncertainty in the estimated mode itself--a key limitation that our framework directly addresses.

\section{Sparse Bayesian PIMs}
\label{sec:sparse-bpim}

Here, we extend the philosophy of Scale-Reliant Inference (SRI) to sparse settings by
treating sparsity as an uncertain but informative property of the data. The key insight,
from Equation~\eqref{eq:rank-1-update}, is that changes in the unidentifiable scale
parameter \(\theta^{\perp}\) induce a uniform shift across all log-fold changes:
\(\theta'=\theta+c\,\mathbf{1}_D\) for some constant \(c\in\mathbb{R}\).
If most entities are unchanged, the correct value of \(\theta^{\perp}\) is the shift that
centers the bulk of identifiable components \(\theta^{\parallel}_1,\dots,\theta^{\parallel}_D\)
at zero. Rather than placing a prior directly on \(\theta^{\perp}\), we estimate it as this
centering shift, propagating uncertainty about the shift in the same way prior SRI methods propagate uncertainty about scale.

Formally, we model log-fold changes as independent draws from a common distribution:
$\theta_d \overset{\text{iid}}{\sim} g, $
where \(g\) is a uniformly continuous density with a unique, finite mode. Potential dependencies among entities (e.g., shared pathways or phylogenetic structure) are absorbed into the shape of \(g\). Letting \(\theta^{\parallel}_d \overset{\text{iid}}{\sim} g^{\parallel}\) denote the identifiable compositional component, we introduce a new definition of sparsity.
\begin{definition}\label{def:sparse-system}
    A system is sparse if
    \setlength{\abovedisplayskip}{5pt}
    \setlength{\belowdisplayskip}{5pt}
    \begin{equation}\label{eq:sparse-system}
        \theta^{\perp} = - \argmax_{t \in [t_l, t_u]} g^{\parallel}(t),
    \end{equation}
    for some arbitrarily wide but finite interval \([t_l, t_u]\).
\end{definition}
The proposed definition is analogous to the \(L_0\) penalty, which regularizes the number of nonzero parameters. This allows us to model the parameter as concentrated around zero without imposing any assumption on the \textit{shape} of this concentration. It can naturally encompass (i) densities with a point mass at zero, (ii) continuous densities that are sharply concentrated around zero, and (iii) continuous densities with a mode at zero but a relatively flat peak--all of which are intuitively `sparse' but have traditionally been treated using different sparse priors or penalties. Consequently, based on Definition~\ref{def:sparse-system}, we can develop models that do not require users to specify hyperparameters or hyperpriors, and that account for uncertainty in the degree of sparsity through uncertainty in the mode of the distribution \(g^{\parallel}\). 

Note that the new definition is more rigorous and yet broader than the qualitative sparsity assumptions common in the literature~\citep{zhou2022linda,grantham2020mimix,srinivasan2020knockoff}: we do not require that most \(\theta_d\)'s are zero, only that the mode of the distribution corresponds to unchanged entities. Even when more than half of the entities are differentially abundant, inference remains valid as long as those effects are sufficiently dispersed that they do not dominate the mode.

\subsection{Mode-Estimating Scale Models}
\label{sec:mode-estimating-sms}

To operationalize our novel sparsity assumption, we extend the Scale Simulation Random Variable (SSRV) framework to estimate \(\theta^{\perp}\) as a function of \(\theta^{\parallel}\). The goal is to choose the shift that aligns the bulk of \(\theta^{\parallel}\) at zero, as implied by the required shift property.

As in previous SSRV applications~(\citealp{gloor2023sri};\citealp{nixon2024a,nixon2024b}), we first generate posterior samples of the identifiable component by fitting \(N\) independent multinomial–Dirichlet models:
\[
W^{\parallel}_{\cdot n} \mid Y_{\cdot n} \sim \text{Dirichlet}(Y_{\cdot n} + \alpha\,\mathbf{1}_D).
\]
For each posterior realization, we compute the corresponding log-fold changes in composition, \(\theta^{\parallel} = (\theta^{\parallel}_1,\dots,\theta^{\parallel}_D)\). A scale model \(p(\theta^{\perp}\mid \theta^{\parallel})\) is then used to generate draws of the unidentifiable scale component, which are combined via the rank-1 update (Equation~\ref{eq:rank-1-update}).
This defines a Bayesian Partially Identified Model (PIM) with a factored prior,
$p(\theta^{\parallel},\theta^{\perp}) = p(\theta^{\parallel})\,p(\theta^{\perp}\mid \theta^{\parallel})$
where \(p(\theta^{\parallel})\) is induced by the Dirichlet priors on \(W^{\parallel}\) (implying \(\mathbb{E}[\theta^{\parallel}]=0\)) and the conditional prior \(p(\theta^{\perp}\mid \theta^{\parallel})\) encodes sparsity.

At this stage, we specify only the expected value of the mode-estimating scale model \(p(\theta^{\perp} \mid  \theta^{\parallel})\) as 
$\mathbb{E}[\theta^{\perp} \mid \theta^{\parallel}] = -\psi(\theta^{\parallel})$
where \(\psi\) is a consistent mode estimator. Intuitively, this chooses the constant shift that moves the mode of \(\theta^{\parallel}\) to zero. We adopt Parzen’s kernel mode estimator~\citep{parzen1962estimation}:
\[
\hat{\theta}^{\perp} = - \argmax_{t\in[t_l,t_u]} p_D(t;\theta^{\parallel}),
\]
for some suitably large interval \([t_{l}, t_{u}]\) 
with a Gaussian kernel density estimator
where the bandwidth \(h(D)\) satisfies standard regularity conditions~\citep{parzen1962estimation}.

Compared to penalized estimator approaches, our approach has two major strengths. First, it avoids the need to choose the size of a penalty, which would otherwise impose assumptions on the degree of sparsity. Instead, the degree of sparsity is automatically determined by the concentration around the mode of \(p_D\). Second, it does not constrain sparsity to a single fixed degree. The mode varies across posterior samples, and consequently the number of nonzero parameters also varies, encoding uncertainty in the degree of sparsity. 

In Section~\ref{sec:consistency}, we show that under these conditions and the sparsity definition in Equation~\eqref{eq:sparse-system}, the resulting Bayes estimator of \(\theta\) is consistent. In Section~\ref{sec:uncertainty}, we extend this by specifying the full distribution \(p(\theta^{\perp}\mid\theta^{\parallel})\), allowing uncertainty about the shift to be explicitly modeled.

\subsection{Consistency of Sparse Bayesian PIMs}
\label{sec:consistency}
\citet{nixon2024b} showed that, without further assumptions, consistent estimation is impossible in SRI. 
However, under the sparsity assumption in Equation~\eqref{eq:sparse-system}, consistency becomes achievable. Our results assume a standard multinomial--Dirichlet measurement model 
for \(Y_{\cdot n}\), a uniformly continuous density \(g\) with a unique mode, 
and a scale model whose conditional mean 
\(\mathbb{E}[\theta^{\perp}\mid\theta^{\parallel}]\) equals the negative mode of Parzen’s kernel density estimator 
\(p_D(t;\theta^{\parallel})\) evaluated over a sufficiently wide interval \([t_l,t_u]\).
Finally, we let \(\lambda=\min(\lambda_{1}, \dots, \lambda_{N})\) denote the minimum sequencing depth afforded by the data. 
(Complete technical conditions and proofs are provided in \ref{sec:lemmas-proofs}.)

\begin{theorem}\label{thm:full_consistency}
Under these assumptions, the Bayes estimator of the Sparse SSRV is consistent:
\[
\bar{\boldsymbol{\theta}}
=\mathbb{E}[\boldsymbol{\theta}\mid Y]
\xrightarrow{p}\boldsymbol{\theta}^{*}
\quad\text{as }\lambda\to\infty \ ,\;D\to\infty \ ,
\]
where \(\boldsymbol{\theta}^{*}\) denotes the true parameter vector.
\end{theorem}

\begin{corollary}\label{cor:posterior_consistency_ssrv}
Assume that, in addition, the posterior variance of the scale model vanishes asymptotically;
$
\mathrm{Var}(\theta^{\perp}\mid\theta^{\parallel})
\xrightarrow[\lambda\to\infty,D\to\infty]{}0.
$
Then the full posterior is consistent: 
\[
\theta_d\mid Y \xrightarrow{p} \theta^{*}_d  
\quad\text{as }\lambda\to\infty \ ,\;D\to\infty \ ,
\quad 
\forall d=1,\ldots,D \ .
\]

\end{corollary}

The variance condition in the Corollary~\ref{cor:posterior_consistency_ssrv} is mild and holds for many practical scale models, including 
the nonparametric model introduced in Section~\ref{sec:uncertainty}, 
where variance shrinks naturally with increasing \(D\). 
As we 
discuss next, deliberately overdispersed models can also be used to improve finite-sample robustness, at the cost of strict posterior consistency.

\subsection{Uncertainty in Mode-Estimating Scale Models}
\label{sec:uncertainty}
We now describe 
the Parzen-based mode-estimating scale model
employed in our proposal.
Recall that 
\(\mathbb{E}[\theta^{\perp}\mid\theta^{\parallel}] = -\psi(\theta^{\parallel})\),
where \(\psi\) is Parzen’s mode estimator.
A degenerate specification,
\(p(\theta^{\perp}\mid\theta^{\parallel})=\delta_{-\psi(\theta^{\parallel})}\),
would ignore uncertainty and yield overconfident inference.
A well-calibrated SSRV must 
reflect uncertainty arising from
(i) dispersion of \(g^{\parallel}\) near its mode,
(ii) limited sample size \(N\),
(iii) finite feature dimension \(D\),
and (iv) finite sequencing depth \(\lambda\).
Only the last of these is propagated by the multinomial--Dirichlet measurement model;
the others must be incorporated through the scale model to avoid
underestimated variance, particularly when sparsity is weak or approximate.

To address this, we adopt a pragmatic approach that combines two complementary
scale models—each centered on Parzen’s mode estimator but differing in how they
quantify additional uncertainty.

\paragraph{1. Laplace-based scale model}  
To account for dispersion in \(g^{\parallel}\) near its mode, we approximate local curvature with a Laplace expansion:
$p(\theta^{\perp}) \;=\; \mathcal{N}(\hat{\theta}^{\perp},\,\tau^2)$
where \(\tau^2\) is the negative inverse Hessian of the estimated density \(p_D(t;\theta^{\parallel})\) evaluated at the estimated mode \(t=\hat{\theta}^{\perp}\). This model reflects how sharply (or diffusely) the density concentrates around its mode. If \(g^{||}\) is smooth at its mode, the Laplace-based scale model yields an overdispersed scale model because the variance does not converge to zero asymptotically.

\paragraph{2. Bootstrap-based scale model}  
To capture variability due to finite \(N\) and finite \(D\), we adopt a two-stage non-parametric bootstrap:  
(i) Resample the \(N\) samples with replacement, recomputing the LFC vector \(\theta^{\parallel}\) to incorporate sample-level uncertainty; and
(ii) resample the \(D\) components of this vector and reapply Parzen’s mode estimator to obtain a bootstrap realization of \(\theta^{\perp}\).  
Repeated resampling propagates uncertainty arising from both the limited number of samples and the finite number of observed features.

In practice, we run both scale models and retain the one with the larger posterior variance for \(\theta^{\perp}\), ensuring conservative inference and stronger Type~I error control when sparsity is uncertain. Pseudo-code is provided in \ref{sec:pseudo_code}. The bootstrap-based model satisfies the shrinking-variance condition of Section~\ref{sec:consistency} and thus achieves full posterior consistency, but combining it with the Laplace-based model yields superior finite-sample performance.

\section{Simulation Studies}
\label{sec:benchmarks-simulation}

We evaluated the Sparse SSRV through simulations varying four factors: the number of features (\(D\)), sample size (\(N\)), sequencing depth (\(\lambda\)), and degree of sparsity. Synthetic datasets were generated using \texttt{SparseDOSSA2}~\citep{ma2021simulation}, which reproduces realistic microbial count distributions based on the Human Microbiome Project (HMP1–2) stool profiles. Case–control contrasts were created by specifying metadata and effect sizes that altered mean absolute abundances between conditions. Detailed simulation parameters and extended results are reported in \ref{sec:sim-benchmark-details}.

\begin{figure}[ht]
    \centering
    \includegraphics[width=1.1\linewidth]{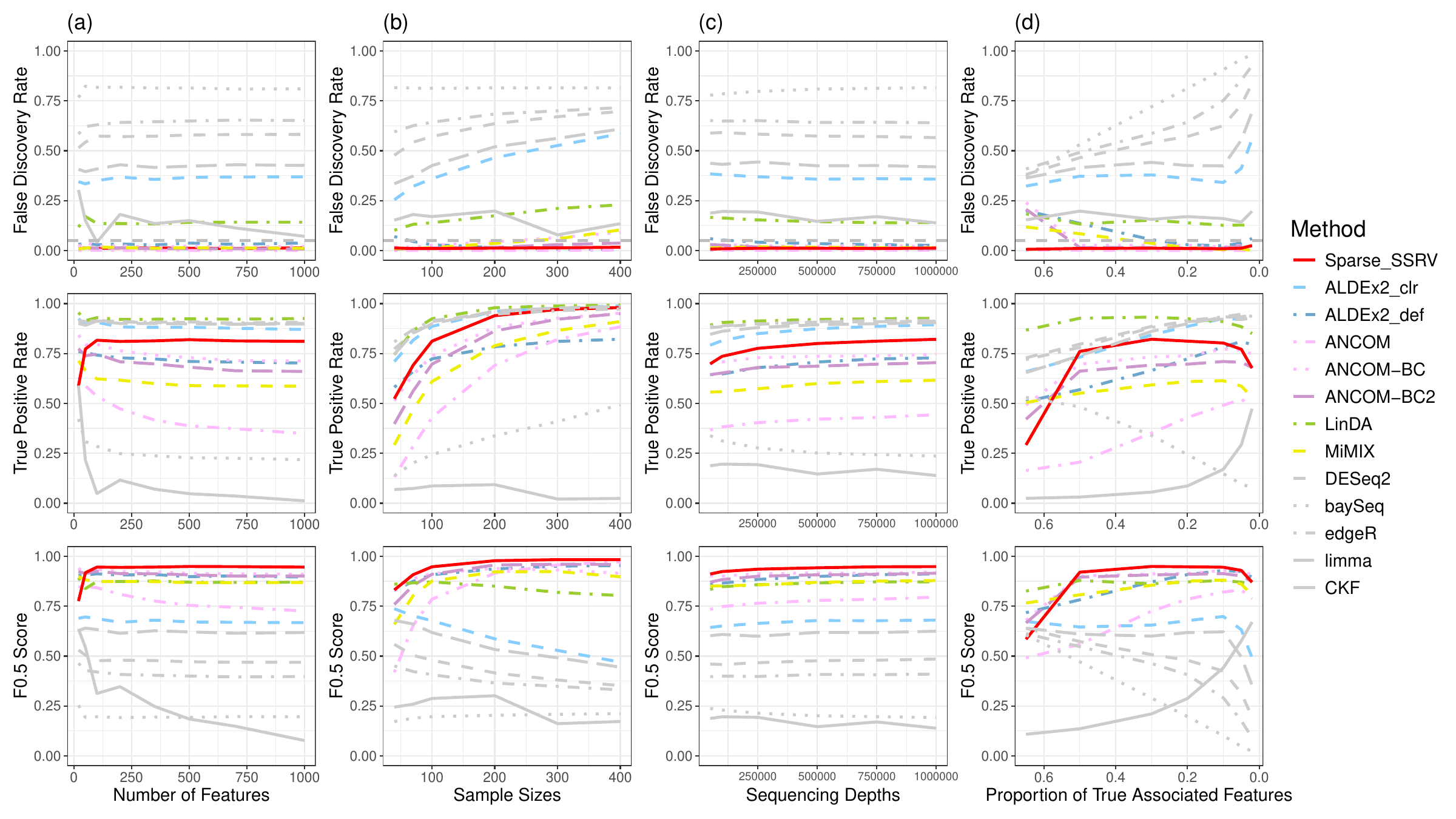}
    \caption{Performance of Sparse SSRV and competing methods using simulated data from SparseDOSSA2 trained on HMP1 and 2 stool profiles. Each panel shows how performance changes with (a) number of features, (b) sample size, (c) sequencing depth, and (d) proportion of true associated features. Baseline values are 300 features, 100 samples, 750,000 reads, and 20\% true associations. In (a), sequencing depth increases linearly with the number of features (2,500 reads per feature) to avoid low-abundance features converging to zero.}
    \label{fig:simulation}
\end{figure}

\begin{figure}[ht]
    \centering
    \includegraphics[width=1\linewidth]{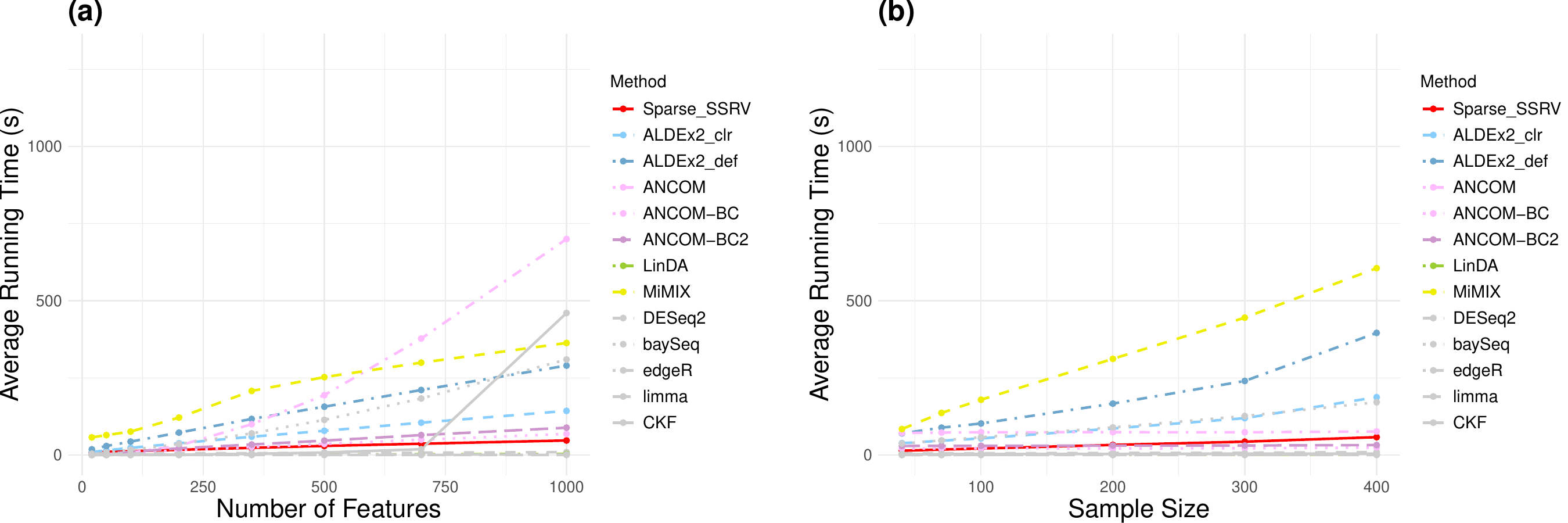}
    \caption{Average runtime of Sparse SSRV and competing methods 
    (color and line type coding match those in Figure~\ref{fig:simulation}). The sample size is fixed at 100 for panel (a), and the number of features is fixed at 300 for panel (b).}
    \label{fig:runtime}
\end{figure}

We compared Sparse SSRV with eleven representative methods: ALDEx2, DESeq2, edgeR, limma, baySeq, MiMIX, ANCOM, ANCOM-BC, ANCOM-BC2, CKF, and LinDA. For ALDEx2, both the standard CLR implementation (\texttt{ALDEx2\_clr}) and its recent stochastic-scale extension (\texttt{ALDEx2\_def}) were evaluated~\citep{nixon2024a}. Methods with extremely high false discovery rate (FDR) or low true positive rate (TPR) appear in gray in Figure~\ref{fig:simulation}; others are highlighted. The \(F_{0.5}\) score summarizes the TPR–FDR trade-off, and its definition is provided in \ref{sec:sim-benchmark-details}.


\paragraph{Number of features}  
Sparse SSRV maintained FDR below 0.05 for all \(D\) while TPR increased with dimensionality, 
in line with the consistency results of Section~\ref{sec:consistency}. Other methods (ANCOM, ANCOM-BC, ANCOM-BC2, ALDEx2\_def, MiMIX) also controlled FDR but lost power as \(D\) grew, indicating reduced efficiency in high dimensions. 


\paragraph{Sample size}  

Most methods gained power with increasing \(N\) but at the cost of inflated FDR. Sparse SSRV preserved strict FDR control and improved TPR as \(N\) grew by adapting uncertainty in \(\theta^{\perp}\) to information in \(\theta^{\parallel}\). In contrast, ALDEx2\_def remained conservative at large \(N\) because its default scale variance does not shrink with sample size.

\paragraph{Sequencing depth}  
Increasing \(\lambda\) improved TPR for all approaches until a plateau was reached. Sparse SSRV consistently maintained FDR below 0.05 across depths, matching its theoretical behavior under increasing \(\lambda\).

\paragraph{Degree of sparsity}  

Sparse SSRV remained robust even when most taxa changed, controlling FDR below 0.05 when 65\% of features were relevant and achieving the highest TPR among methods with acceptable FDR. This demonstrates that the ``sparse system'' definition (Section~\ref{sec:sparse-bpim}) is more flexible than the strict sparsity assumptions used elsewhere.


\paragraph{Computational efficiency}
Sparse SSRV scaled efficiently with both \(D\) and \(N\), matching or outperforming comparably accurate methods in runtime (Figure~\ref{fig:runtime}); only LinDA was faster at very large \(D\).


\section{Real Data Studies}
\label{sec:benchmarks-real}

We evaluated Sparse SSRV, along with competitor methods, on six published datasets spanning diverse hosts, sample types, and sparsity structures \citep{mcmurrough2014selex,zemb2020absolute,stammler2016adjusting,jin2022quantsinglecell,vandeputte2017quantitative,prasse2015soil}.
Key dataset characteristics are summarized in Table~\ref{tab:real_summary}.
For datasets lacking direct ground truth (Jin, Vandeputte, Prasse), we followed
\citet{nixon2024b} and used ALDEx2 with a measurement-based scale model
as the reference (details can be found in \ref{sec:detail_real}).
These six datasets span diverse hosts, sample types, and
degrees of sparsity, enabling a broad evaluation of the comparative performance of Sparse SSRV under realistic conditions. 
Results 
for all datasets are summarized in Table~\ref{tab:real_results}.

\begin{table}[ht]
    \centering
    \resizebox{1\textwidth}{!}{
    \begin{tabular}{c c c c c c c c}
    \toprule
        Scenario & Dataset & Sample & \shortstack{Dimension \\ (D$\times$N)} & Reference & \shortstack{Sequencing \\ method} & Truth Definition & \shortstack{Signal Sparsity \\ (Truths / D)} \\
        \cmidrule(r){1-1}\cmidrule(lr){2-2}\cmidrule(l){3-3}\cmidrule(l){4-4}\cmidrule(l){5-5}\cmidrule(l){6-6}\cmidrule(l){7-7}\cmidrule(l){8-8}
        \multirow{4}{*}{\shortstack{ \\ \\ Sparse \\ Low Variance}} & \shortstack{McMurrough et al \\ (2014)} & \shortstack{} & 1600 $\times$ 14 & \cite{mcmurrough2014selex} & 16s rRNA & Known Truth & 0.017\\
        & \shortstack{Zemb et al \\ (2020)} & \shortstack{Pig \\ (fecal)} & 428 $\times$ 10 & \cite{zemb2020absolute} & 16s rRNA & Known Truth & 0.002\\
        & \shortstack{Stämmler et al \\ (2016)} & \shortstack{Mice \\ (cecum content)} & 44 $\times$ 12 & \cite{stammler2016adjusting} & 16s rRNA & Known Truth & 0.048\\
        & \shortstack{Jin et al \\ (2022)} & \shortstack{Mice \\ (cecum content)} & 132 $\times$ 16 & \cite{jin2022quantsinglecell} & 16s rRNA & Gold Standard Model & 0.005 \\
        \midrule
        \shortstack{Sparse \\ High Variance} & \shortstack{Prasse et al \\ (2015)} & \shortstack{Soil} & 181 $\times$ 57 & \cite{prasse2015soil} & 16s rRNA & Gold Standard Model & 0.119 \\
        \midrule
        \shortstack{Not Sparse \\ High Variance} & \shortstack{Vandeputte et al \\ (2017)} & \shortstack{Human \\ (fecal)} & 29 $\times$ 95 & \cite{vandeputte2017quantitative} & 16s rRNA & Gold Standard Model & 0.793 \\
        \bottomrule
    \end{tabular}
}
    \caption{Characteristics of the six datasets used to evaluate differential analysis (DA) performance.
    Datasets are categorized as \textit{sparse/low variance}, \textit{sparse/high variance}, or \textit{not sparse}. Here, ``sparse'' indicates that the underlying system satisfies the sparse system assumption utilized in Sparse SSRV. Variance refers to variability in the LFCs of null features, which affects the distinctness of the mode in the LFC density. For \citet{mcmurrough2014selex}, the sample type is omitted because data were generated in vitro via SELEX.}
    \label{tab:real_summary}
\end{table}

\begin{table}[ht]
    \renewcommand{\arraystretch}{1.5}  
\resizebox{0.72\textwidth}{!}{
\begin{minipage}{.1\linewidth}
  \centering
  \vspace{2.52cm}
  \scalebox{1}[1]{
  \begin{tabular}{c c}
    \toprule
    \multirow{4}{*}{\shortstack{ \\ \\ Sparse \\ Low Variance}} & \shortstack{McMurrough \\ (2014)} \\
    & \shortstack{Zemb \\ (2020)} \\
    & \shortstack{Stämmler \\ (2016)} \\
    & \shortstack{Jin \\ (2022)} \\ \hline
    \shortstack{Sparse \\ High Variance} & \shortstack{Prasse \\ (2015)} \\ \hline
    \shortstack{Not Sparse} & \shortstack{Vandeputte \\ (2017)} \\
    \bottomrule
  \end{tabular}
  }
\end{minipage}%
\hspace{4.1cm}
\begin{minipage}{.1\linewidth}
  \centering
  \vspace{1.45cm}
  \scalebox{1}[1.06]{
  \begin{tabular}{c}
    \shortstack{D x N \\ (Truth)} \\
    \midrule
    \shortstack{1600 x 14 \\ (27)} \\
    \shortstack{428 x 10 \\ (1)} \\
    \shortstack{44 x 12 \\ (2)} \\
    \shortstack{132 x 16 \\ (1)} \\ \hline
    \shortstack{181 x 57 \\ (10)} \\ \hline
    \shortstack{29 x 95 \\ (23)} \\
    \bottomrule
  \end{tabular}
  }
\end{minipage}%
\hspace{.5cm}
\begin{minipage}{.57\linewidth}
  \centering
  \vspace{-0.5cm}
  \renewcommand{\arraystretch}{1.5}
  \scalebox{1}[1.06]{
  \begin{tabular}{c c c c c c c c c c c c}
    \rotatebox{75}{Sparse} & 
    \rotatebox{75}{ALDEx2\_def} & 
    \rotatebox{75}{DESeq2} & 
    \rotatebox{75}{edgeR} & 
    \rotatebox{75}{limma} & 
    \rotatebox{75}{baySeq} & 
    \rotatebox{75}{MiMIX} & 
    \rotatebox{75}{ANCOM-BC2} & 
    \rotatebox{75}{ANCOM-BC} & 
    \rotatebox{75}{ANCOM} & 
    \rotatebox{75}{LinDA} &
    \rotatebox{75}{CKF}\\
    \midrule
    \shortstack{\textbf{9} \\ \textbf{(10)}} & 
    \shortstack{15 \\ (9)} & 
    \shortstack{577 \\ (21)} & 
    \shortstack{320 \\ (25)} & 
    \shortstack{93 \\ (7)} & 
    \shortstack{364 \\ (9)} & 
    \shortstack{541 \\ (24)} & 
    \shortstack{62 \\ (10)} & 
    \shortstack{1314 \\ (18)} & 
    \shortstack{29 \\ (9)} & 
    \shortstack{999 \\ (22)} &
    \shortstack{0 \\ (0)}\\ 
    \shortstack{\textbf{0} \\ \textbf{(1)}} & 
    \shortstack{0 \\ (0)} & 
    \shortstack{3 \\ (1)} & 
    \shortstack{2 \\ (1)} & 
    \shortstack{1 \\ (1)} & 
    \shortstack{9 \\ (1)} & 
    \shortstack{0 \\ (0)} & 
    \shortstack{1 \\ (0)} & 
    \shortstack{37 \\ (1)} & 
    \shortstack{0 \\ (0)} & 
    \shortstack{1 \\ (0)} &
    \shortstack{0.76 \\ (0.76)} \\ 
    \shortstack{\textbf{0} \\ \textbf{(2)}} & 
    \shortstack{\textbf{0} \\ \textbf{(2)}} & 
    \shortstack{\textbf{0} \\ \textbf{(2)}} & 
    \shortstack{\textbf{0} \\ \textbf{(2)}} & 
    \shortstack{\textbf{0} \\ \textbf{(2)}} & 
    \shortstack{2 \\ (0)} & 
    \shortstack{0 \\ (1)} & 
    \shortstack{\textbf{0} \\ \textbf{(2)}} & 
    \shortstack{1 \\ (2)} & 
    \shortstack{\textbf{0} \\ \textbf{(2)}} & 
    \shortstack{2 \\ (2)} &
    \shortstack{\textbf{0} \\ \textbf{(2)}} \\ 
    \shortstack{4 \\ (1)} & 
    \shortstack{\textbf{0} \\ \textbf{(1)}} & 
    \shortstack{14 \\ (1)} & 
    \shortstack{22 \\ (1)} & 
    \shortstack{5 \\ (1)} & 
    \shortstack{12 \\ (0)} & 
    \shortstack{3 \\ (0)} &
    \shortstack{\textbf{0} \\ \textbf{(1)}} & 
    \shortstack{10 \\ (1)} & 
    \shortstack{\textbf{0} \\ \textbf{(1)}} & 
    \shortstack{17 \\ (1)} &
    \shortstack{3.6 \\ (1)} \\ \hline
    \shortstack{\textbf{0} \\ \textbf{(1)}} & 
    \shortstack{3 \\ (0)} & 
    \shortstack{61 \\ (2)} & 
    \shortstack{10 \\ (1)} & 
    \shortstack{16 \\ (10)} & 
    \shortstack{14 \\ (0)} & 
    \shortstack{1 \\ (0)} & 
    \shortstack{11 \\ (1)} & 
    \shortstack{82 \\ (3)} & 
    \shortstack{0 \\ (0)} & 
    \shortstack{18 \\ (6)} &
    \shortstack{0.57 \\ (0.45)} \\ \hline
    \shortstack{3 \\ (0)} & 
    \shortstack{3 \\ (5)} & 
    \shortstack{6 \\ (16)} & 
    \shortstack{6 \\ (13)} & 
    \shortstack{6 \\ (6)} & 
    \shortstack{\textbf{4} \\ \textbf{(16)}} & 
    \shortstack{2 \\ (4)} & 
    \shortstack{2 \\ (6)} & 
    \shortstack{5 \\ (14)} & 
    \shortstack{1 \\ (3)} & 
    \shortstack{6 \\ (7)} &
    \shortstack{2.06 \\ (3.45)} \\ 
    \bottomrule
  \end{tabular}
  }
\end{minipage}
}

    \caption{Summary of DA results across six
    datasets. Each cell reports the number of false positives, with true positives in parentheses. Methods with the best false-to-true positive ratios are bolded. CKF results are averaged over 100 runs due to high variability.}
    \label{tab:real_results}
\end{table}

\paragraph{Sparse system with low variance
(ideal setting)}  
Four datasets (McMurrough, Zemb, Stämmler, Jin) display sparse signals with low
variance among null features.
Sparse SSRV achieved the lowest false-to-true positive ratios in three of the four,
recovering more true positives than competitors and maintaining perfect FDR control.
In McMurrough and Zemb, where total scale differed markedly between conditions,
Sparse SSRV identified all true positives without false discoveries;
in Stämmler and Jin, where scales were similar, several methods performed comparably,
but Sparse SSRV remained competitive in both FDR and power.

\paragraph{Sparse system with high variance (challenging setting)}  
The Prasse dataset exhibits dispersed null effects and no clear LFC mode.
Here, all methods struggled; Sparse SSRV detected one true positive but no false
positives, demonstrating its conservative behavior.
Methods such as limma, ANCOM-BC, and LinDA recovered more true positives but
showed inflated FDR elsewhere, whereas Sparse SSRV maintained error control
across all settings.

\paragraph{Dense system (violation of the sparse system assumption)} 
Vandeputte contains dense signals, with most features differential and the LFC mode located among non-null effects. As expected, Sparse SSRV was less powerful but maintained moderate FDR through its high-variance scale model. Methods optimized for dense signals (DESeq2, edgeR, baySeq, ANCOM-BC) achieved higher power but at the expense of inflated FDR in sparse systems.

\vspace{.1in}
\noindent
These results highlight a key strength of Sparse SSRV: its performance is both robust and interpretable. Unlike many other methods, it rests on a clearly \emph{diagnosable assumption}--the sparse system condition--which can be evaluated directly by inspecting the estimated LFC density. 
When this assumption holds, Sparse SSRV attains strong power with tight Type~I error control; when it fails, the model’s built-in uncertainty yields conservative, well-calibrated inference. 
This transparency makes Sparse SSRV a reliable tool for differential abundance analysis across diverse settings.

\section{Discussion}
\label{sec:discussion}
In this paper, we introduced \textit{Sparse SSRV}, a sparse Bayesian Partially Identified Model (PIM) that extends the Scale-Reliant Inference (SRI) framework to settings where only a subset of features change between conditions. Unlike normalization-based methods and recent sparse methods, which achieve \textit{artificial identifiability} by fixing scale to a single (often implausible) value, Sparse SSRV provides a more flexible and biologically plausible treatment of sparsity. By explicitly modeling uncertainty in sparsity, it captures a wide spectrum of sparsity degrees and, in doing so, also accounts for uncertainty in scale. Propagating this uncertainty through calculations, Sparse SSRV provides a principled alternative to approaches that have been shown to fail when their assumptions are even modestly violated.

Our framework formalizes sparsity through the notion of a \emph{sparse system}, where the mode of the log-fold change (LFC) distribution corresponds to unchanged features. This definition is weaker and more biologically plausible than traditional qualitative sparsity assumptions that discuss ``most taxa being unchanged'', and it can be easily diagnosed empirically by inspecting the estimated LFC density. Within this framework, we introduced a mode-estimating scale model and proved that the resulting Sparse SSRV is consistent under mild conditions. 

Extensive simulations and real-data analyses demonstrate that Sparse SSRV performs reliably across diverse settings. When the sparse system assumption holds, it achieves lower false discovery rates (FDR) and higher power than competing methods, including other Bayesian and penalized sparse models. Even when the assumption is violated—such as in dense systems—Sparse SSRV fails gracefully, maintaining conservative FDR control rather than producing spurious discoveries. This transparency and robustness make Sparse SSRV a practical and interpretable tool for applied researchers. Moreover, the method is computationally efficient, scaling well to high-dimensional sequence count data. 
While differential abundance in microbiome studies was employed as our ``running example'' throughout the paper, Sparse SSRV is immediately applicable to differential expression in gene transcription studies (see Introduction), and could be a valuable tool in a variety of other applications.

Despite these strengths, Sparse SSRVs have limitations. In particular, they may exhibit reduced power when the LFCs of null features are highly variable, resulting in a diffuse or multimodal distribution. In such cases, the mode may not clearly correspond to the set of unchanged features, which undermines the sparse system assumption. One possible remedy is to incorporate external information—such as known control features, technical spike-ins, or phylogenetic structure—to better regularize the mode estimate. Another strategy is to modify the scale model to explicitly allow for multimodality, for instance by averaging over a mixture of local modes or by upweighting the central region of the LFC distribution.

Additionally, our development focused on binary condition variables. However, many biological studies involve multiple conditions or continuous exposures. Extending Sparse SSRV to a regression framework would allow the mode-estimating scale model to generalize naturally to settings with covariates. In future work, we also plan to incorporate hierarchical structures (e.g., subject-level random effects) and temporal or spatial
correlations. These extensions would enable robust and interpretable differential analysis in more complex experimental designs, while retaining the core principle of propagating uncertainty in scale and sparsity.

Overall, Sparse SSRV unifies and improves upon two previously disconnected lines of work: normalization-based SRI methods and sparse penalized estimators. By embracing, rather than concealing, the uncertainty inherent in both scale and sparsity, it provides a robust and interpretable framework for differential analysis under scale uncertainty.


\backmatter

\bibliographystyle{biom} 
\bibliography{references}

\label{lastpage}

\setcounter{section}{0}
\renewcommand{\thesection}{Supplementary Section \arabic{section}}

\section{Rank-1 Update for General Log-linear Estimands}\label{sec:general-rank-1} 
Section~\ref{sec:sri-review} introduces the rank-1 update formulation for log-fold-changes (LFCs). However, this formulation can be extended to other estimands, even beyond the context of differential abundance analysis. In fact, any estimand that is a linear functional of log abundances (\(W\)) can be written as a rank-1 update, as it can be decomposed with respect to composition and scale:
\setlength{\abovedisplayskip}{5pt plus 1pt minus 1pt}
\setlength{\belowdisplayskip}{5pt plus 1pt minus 1pt}
\begin{equation}
    \theta = \varphi[f(W)] = \varphi[f(W^{||})+f(\mathbf{1}_DW^{\perp})]=\underbrace{\varphi[f(W^{||})]}_{\theta^{||}}+\underbrace{\varphi[f(\mathbf{1}_DW^{\perp})]}_{\theta^\perp\mathbf{1}_D}
    \label{eq:general_rank-1-update}
\end{equation}
where \(\varphi\) is an arbitrary linear functional and \(f\) is any logarithmic map (i.e., any map such that \(f(ab)=f(a)+f(b)\)). While our main text focuses on LFC estimation, this simple manipulation shows that the same theory and methods can be extended to any linear estimands, including the least squares estimator in regression models. An interesting example of this rank-1 update framework would be function smoothing. Consider a basis function regression where the estimand of interest is the vector of smoothing coefficients \(\mu_{d}\) of a sparsely and noisily observed function \(\log W_{d.}= (\log W_{d1}, \dots, \log W_{dN})\) expressed with respect to a basis expansion $\log W_{dn} = \sum_{j=1}^J \mu_{dj} e_j(n) + \epsilon_{n}$. Let $E$ denote the basis matrix, with entries $E_{nj} = e_j(x_n)$. Then the ordinary least squares estimator is 
 \begin{equation*}
     \hat \mu_d = (E^\top E)^{-1} E ^\top \log W_{d \cdot} = \hat \mu^{||}_d + \hat \mu^\perp \mathbf{1_D},
 \end{equation*}
where $\hat \mu^{||}_d = (E^\top E)^{-1} E ^\top \log W^{||}_{d\cdot} \text{ and } \hat \mu ^\perp = (E^\top E)^{-1} E ^\top \log W^{\perp},$ allowing rank-1 update of $\hat \mu$ via $\hat \mu ^\perp$. This rank-1 update framework also extends naturally to penalized variants, such as smoothing splines and ridge regressions. Thus, while our primary focus is on LFCs, the rank-1 update formulation is applicable to a much broader class of estimands, suggesting that the method proposed in the main article can be extended to a variety of analyses beyond differential abundance.

\section{Example—ALDEx2 as a Special Case of SSRVs}\label{sec:aldex-example}
To illustrate the concept of Scale Simulation Random Variables (SSRVs), we revisit the ALDEx2 procedure~\citep{fernandes2014unifying}, a widely used method for differential abundance analysis. ALDEx2 can be viewed as a degenerate instance of an SSRV in which the unidentifiable scale parameter \(\theta^{\perp}\) is fixed deterministically rather than modeled stochastically.

\paragraph{Measurement model}
ALDEx2 begins by estimating compositional uncertainty through independent
Dirichlet--multinomial models:
\[
W^{\parallel}_{\cdot n}\mid Y_{\cdot n} \sim 
\text{Dirichlet}(Y_{\cdot n}+\alpha\mathbf{1}_D),
\qquad \alpha = 0.5.
\]
Each posterior draw \(W^{\parallel}_{\cdot n}\) represents a plausible composition for sample
\(n\), conditional on its observed counts \(Y_{\cdot n}\).

\paragraph{Implicit scale assumption (CLR normalization)}
ALDEx2 then applies a centered log-ratio (CLR) transformation,
\[
\log \tilde{W}_{\cdot n} = 
\log W^{\parallel}_{\cdot n} - 
\phi(W^{\parallel}_{\cdot n}), \qquad 
\phi(W^{\parallel}_{\cdot n}) = \frac{1}{D}\sum_{d=1}^{D}\log W^{\parallel}_{dn}.
\]
This step implicitly assumes a deterministic relationship between composition and scale:
\[
\log W^{\perp}_{n} = -\phi(W^{\parallel}_{\cdot n}),
\]
and therefore fixes the scale component of the log-fold change to
\[
\theta^{\perp}_{\mathrm{CLR}}
= -\mean_{n:x_n=1}\phi(\log W^{\parallel}_{\cdot n})
  +\mean_{n:x_n=0}\phi(\log W^{\parallel}_{\cdot n}).
\]
Under this \emph{CLR assumption}, ALDEx2 corresponds to an SSRV with a degenerate scale model \(p(\theta^{\perp}\mid \theta^{\parallel}) = \delta_{\theta^{\perp}_{\mathrm{CLR}}}\).

\paragraph{Stochastic scale model}
Replacing the fixed normalization with a stochastic model for total load yields a \emph{proper} SSRV:
\[
\log W^{\perp}_{n}\sim 
N\!\left(-\phi(\log W^{\parallel}_{\cdot n}),\,\gamma^{2}\right),
\qquad
\log W_{\cdot n} = 
\log W^{\parallel}_{\cdot n} + 
\log W^{\perp}_{n}\,\mathbf{1}_{D}.
\]
Here, \(\gamma^{2}\) quantifies prior uncertainty about scale differences across samples. Sampling \(W^{\perp}_{n}\) from this model propagates scale uncertainty through to the estimand \(\theta\), complementing the compositional uncertainty in \(W^{\parallel}\).

\paragraph{Empirical implications}
In practice, this stochastic extension preserves the computational simplicity of ALDEx2 while yielding posteriors \(p(\theta\mid Y)\) that better reflect the information content of the data. Simulations and reanalyses of real datasets have shown that even modest variance in the scale prior (\(\gamma^{2}\approx 0.25\)) can reduce false discovery rates by more than half without sacrificing power~\citep{nixon2024a,nixon2024b,mcgovern2024intervals}.

\smallskip
\noindent
This example illustrates how deterministic normalizations correspond to degenerate SSRVs, and how introducing stochastic scale models restores proper Bayesian treatment of non-identifiable scale parameters.

\section{Pitfall of Sum-to-zero Constraint}
\label{sec:ckf-sim}
Section~\ref{sec:pitfall-normalized-sparsity} of the main article discusses how the sum-to-zero constraint implies symmetry between gains and losses in abundance across taxa. To quantify the impact of this unrealistic assumption, we conducted a simulation study using synthetic sparse datasets and analyzed them with the Compositional Knockoff Filter (CKF), which is based on Equation~\eqref{eq:penalized-estimator}~\citep{srinivasan2020knockoff}. 

We used SparseDOSSA2~\citep{ma2021simulation} to generate synthetic datasets consisting of 400 features and 100 samples, with 20 features (5\%) assigned nonzero effect sizes. These effect sizes were randomly drawn from the distribution $5\times beta(1,3)+1$, resulting in values ranging from 1 to 6. Additionally, the sequencing depth is set to $2\cdot 10^7$ so these effect sizes are sufficiently realized in the simulated data.  We compared the performance of CKF under varying ratios of negative to positive effect sizes: 10:10 (Symmetric), 4:16 (80\% Positive), or 0:20 (All Positive). For each scenario, we simulated 1000 datasets and analyzed each twice using CKF with two different target FDR levels (0.05 and 0.1) and also Sparse SSRV. \ref{tab:ckf_fdr_power_tab_0.05} displays the average empirical FDR and power achieved by CKF and Sparse SSRV under each scenario.

\begin{table}[!ht]
\centering
\begin{tabular}[t]{l|l|l|l}
\hline
  & Symmetric & 80\% Positive & All Positive\\
\hline
CKF target FDR 0.1 & 0.07 / 0.23 & 0.12 / 0.21 & 0.21 / 0.18\\
\hline
CKF target FDR 0.05 & 0.04 / 0.21 & 0.09 / 0.19 & 0.16 / 0.17\\
\hline
Sparse SSRV & 0.03 / 0.85 & 0.03 / 0.86 & 0.03 / 0.87\\
\hline
\end{tabular}
\caption{Summary of empirical FDR and power across methods and scenarios. Each cell reports the average FDR and power, separated by a slash (FDR / Power), computed over 1000 simulations.}
\label{tab:ckf_fdr_power_tab_0.05}
\end{table}

\ref{tab:ckf_fdr_power_tab_0.05} shows that the performance of CKF deteriorates as the sum-to-zero constraint is increasingly violated. When the target FDR is 0.05, CKF achieves an empirical FDR of 0.04 in the Symmetric case. However, this inflates to 0.09 in the 80\% Positive case and rises further to 0.16 in the All Positive case. Notably, despite the increase in FDR, power decreases as the ratio of positive and negative effect sizes become more asymmetric. In contrast, Sparse SSRV not only show a much higher power than CKF, but also exhibits stable FDR control and power across all cases. 

\begin{table}[!ht]
\centering
\begin{tabular}[t]{l|l|l|l}
\hline
  & Symmetric & 80\% Positive & All Positive\\
\hline
CKF target FDR 0.1 & 0.02 / 0.07 & 0.06 / 0.06 & 0.21 / 0.05\\
\hline
CKF target FDR 0.05 & 0.01 / 0.06 & 0.04 / 0.05 & 0.19 / 0.04\\
\hline
Sparse SSRV & 0.02 / 0.90 & 0.02 / 0.90 & 0.02 / 0.91\\
\hline
\end{tabular}
\caption{Summary of empirical FDR and power across methods and scenarios. Each cell reports the average FDR and power, separated by a slash (FDR / Power), computed over 1000 simulations.}
\label{tab:ckf_fdr_power_tab_0.20}
\end{table}

We repeated the simulation with a more moderate degree of sparsity, where 20\% features were assigned nonzero effect sizes. \ref{tab:ckf_fdr_power_tab_0.20} shows a similar pattern as \ref{tab:ckf_fdr_power_tab_0.05}: FDR increases as the effect sizes become more asymmetric. However, CKF achieved substantially lower power compared to the previous simulation with more extreme sparsity. In particular, under the `All Positive' scenario, CKF exhibited a similar FDR to the previous simulation but with much lower power, indicating that the model is more biased in the setting with moderate sparsity.

Our results show that even modest violations of the sum-to-zero assumption lead to severe inflation of false discovery rates, supporting our theoretical argument in the main article. Furthermore, these simple experiments reveal how the performance of a method can depend on the degree of sparsity it implicitly assumes. In both sparsity settings, most of the entities remain unchanged. However, CKF achieved substantially lower power in the second simulation with moderate sparsity, while still exhibiting a similar degree of inflated FDR under severe violation of the sum-to-zero constraint. These results show that CKF performs well only within a certain range of sparsity, not necessarily covering all contexts that are intuitively sparse. 


\section{Lemmas and Proofs}
\label{sec:lemmas-proofs}

\begin{lemma} \label{lem:consistency_W}
    Assume a likelihood \(Y_{\cdot n} \sim \text{Multinomial}(\lambda_{n}, W^{\parallel}_{\cdot n})\) and prior $W^{\parallel}_{\cdot n} \sim \mathrm{Dir}(\alpha \boldsymbol{1_D})$ where \(0<\alpha<\infty\). Let \(\lambda = \inf_{n \in \{1, \dots, N\}} \lambda_{n}\) denote the minimum sequencing depth over the \(N\) samples. Then, 
    \begin{equation}
        W^{\parallel}_{dn}|Y_{\cdot n} \xrightarrow{p} W^{\parallel,*}_{dn} \hspace{0.5em}\text{as} \hspace{0.5em} \lambda \rightarrow \infty
    \end{equation}
    for any \(d \in \{1,2,\dots,D\}\) and where \(W^{\parallel,*}_{dn}\) denotes the true value of \(W^{\parallel}_{dn}\).
\end{lemma}

\begin{proof}
Conjugacy gives the posterior
\[
W^{\parallel}_{dn}\mid Y_{\cdot n}\sim \text{Dirichlet}(\alpha+Y_{\cdot n}),
\]
with mean and variance
\[
E[W^{\parallel}_{dn}\mid Y_{\cdot n}]=\frac{\alpha+Y_{dn}}{D\alpha+\lambda_n},
\qquad
\text{Var}(W^{\parallel}_{dn}\mid Y_{\cdot n})
=\frac{(\alpha+Y_{dn})(\lambda_n+D\alpha-(\alpha+Y_{dn}))}{(D\alpha+\lambda_n)^2(D\alpha+\lambda_n+1)}.
\]
By the law of large numbers for multinomials, \(Y_{dn}/\lambda_n\xrightarrow{\text{a.s.}}W^{\parallel,*}_{dn}\), so
\[
E[W^{\parallel}_{dn}\mid Y_{\cdot n}] \to W^{\parallel,*}_{dn}.
\]
The variance satisfies \(\text{Var}(W^{\parallel}_{dn}\mid Y_{\cdot n})=O(1/\lambda_n)\). By Chebyshev’s inequality,
\[
P\bigl(|W^{\parallel}_{dn}\mid Y_{\cdot n}-W^{\parallel,*}_{dn}|>\varepsilon\bigr)
\le \frac{\text{Var}(W^{\parallel}_{dn}\mid Y_{\cdot n})}{\varepsilon^2}
\to 0,
\]
establishing convergence in probability.
\end{proof}

Lemma~\ref{lem:consistency_W} ensures posterior samples $W^{\parallel}_{dn}|Y_{\cdot n}$ contract around the true $W^{\parallel}_{dn}$ as $\lambda$ increases. Then, by continuous mapping theorem, we immediately get the following result. 

\begin{lemma} \label{lem:consistency_theta_pll}
    Define a randomized estimator \(\hat{\theta}^{\parallel}=(\hat{\theta}^{\parallel}_{1},\dots ,\hat{\theta}^{\parallel}_{D})\) as the posterior \(p(W^{\parallel}\mid  Y)\) defined above transformed via :
    \begin{equation}
    \hat{\theta}^{\parallel}_{d} = \underset{n:x_{n}=1} {\mathrm{mean}}\log (W^{\parallel}_{dn}|Y_{\cdot n}) - \underset{n:x_{n}=0} {\mathrm{mean}} \log (W^{\parallel}_{dn}|Y_{\cdot n}).
    \end{equation}
    Then
    \begin{equation}
        \hat{\theta}^{\parallel}_{d} \xrightarrow{p} \theta^{\parallel}_{d} \hspace{0.5em} \mathrm{as} \hspace{0.5em} \lambda \rightarrow \infty
    \end{equation}
    for every \(d \in \{1,2,...,D\}\).
\end{lemma}
\begin{proof}
The mapping \(W^{\parallel}_{d} \mapsto \hat{\theta}^{\parallel}_{d}\) is continuous. Lemma~\ref{lem:consistency_W} and the continuous mapping theorem give the result.
\end{proof}

Using Lemma~\ref{lem:consistency_theta_pll}, we show that the 
Parzen mode estimator-based scale model produces a consistent Bayes estimator for $\theta^{\perp}$.

\begin{lemma}\label{lem:consistency_mode}
Assume $g$, the true density of $\theta^{\parallel}_d$, is uniformly continuous with a unique mode. 
Let $p_D(t;\hat{\theta}^{\parallel})$ be Parzen’s kernel density estimator with Gaussian kernel and bandwidth $h(D)$ satisfying regularity conditions : $\lim_{D\to\infty} h(D)=0,\quad \lim_{D\to\infty} D\,h(D)^2=\infty.$
Define the Bayes estimator of the scale parameter as the posterior mean:
\[
\bar{\theta}^{\perp} \;=\; \mathbb{E}[\theta^{\perp}\mid Y],
\text{ with } \mathbb{E}[\theta^{\perp}\mid\theta^{\parallel}]
= -\,\argmax_{t\in[t_l,t_u]} p_D(t;\theta^{\parallel})
\]
for any suitably large interval \([t_{u}, t_{u}]\) containing \(\theta^{\perp}\). 
Then, as $\lambda\to\infty$ and subsequently $D\to\infty$,
\[
\bar{\theta}^{\perp}
\;\xrightarrow{p}\;
-\argmax_{t\in[t_l,t_u]} g(t)
\;=\;\theta^{\perp,*}
\]
where \(\theta^{\perp,*}\) denotes the true value of \(\theta^{\perp}\).
\end{lemma}
\begin{proof}
First, by Lemma~\ref{lem:consistency_W} and Lemma~\ref{lem:consistency_theta_pll}, 
\(\hat{\theta}^{\parallel}\xrightarrow{p}\theta^{\parallel}\) as $\lambda\to\infty$ for each fixed $D$.

Because $p_D(t;\eta)$ is jointly continuous in $(t,\eta)$ and has a unique maximizer, the Maximum theorem implies the functional 
\[
\phi(\eta)\equiv -\argmax_{t\in[t_l,t_u]} p_D(t;\eta)
\]
is continuous. Thus, by the continuous mapping theorem,
\[
\phi(\hat{\theta}^{\parallel})
\xrightarrow[\lambda\to\infty]{p}
-\argmax_{t\in[t_l,t_u]} p_D(t;\theta^{\parallel}).
\]
Finally, by Theorem 3A of Parzen~\cite{parzen1962estimation}, $p_D(t;\theta^{\parallel})\to g(t)$ uniformly as $D\to\infty$, implying
\[
-\argmax_{t\in[t_l,t_u]} p_D(t;\theta^{\parallel})
\to -\argmax_{t\in[t_l,t_u]} g(t)
= \theta^{\perp}.
\]
Since $\mathbb{E}[\theta^{\perp}\mid Y]$ equals the deterministic functional $\phi(\hat{\theta}^{\parallel})$ by construction, the result follows.
\end{proof}

Lemma~\ref{lem:consistency_mode} states that the posterior mean of the scale model integrated over the measurement model \(p(\theta^{\perp} \mid  Y)=\int p(\theta^{\perp} \mid  \theta^{\parallel})p(\theta^{\parallel} \mid  Y)\) is a consistent estimator for \(\theta^{\perp}\). Under the mild additional assumption that the variance of the scale model shrinks asymptotically, then this result is strengthened and the entire posterior concentrates about the true value \(\theta^{\perp,*}\).

\begin{corollary}\label{cor:posterior_consistency_theta_perp}
Under the conditions of Lemma~\ref{lem:consistency_mode}, 
suppose further that the posterior variance of the scale model vanishes:
\[
\mathrm{Var}(\theta^{\perp}\mid\theta^{\parallel}) \xrightarrow[\lambda\to\infty,\;D\to\infty]{} 0.
\]
Then the posterior for $\theta^{\perp}$ is consistent:
\[
\theta^{\perp}\mid Y \xrightarrow{p} \theta^{\perp}.
\]
\end{corollary}

\begin{proof}
By Lemma~\ref{lem:consistency_mode}, 
$\mathbb{E}[\theta^{\perp}\mid Y]\xrightarrow{p}\theta^{\perp}$. 
Chebyshev’s inequality yields
\[
P(|\theta^{\perp}-\theta^{\perp,*}|>\varepsilon\mid Y)
\le\frac{\mathrm{Var}(\theta^{\perp}\mid Y)}{\varepsilon^2}
\to 0,
\]
since, by the law of total variance and the conditional independence \(\theta^{\perp} \perp Y \mid  \theta^{\parallel}\), $\mathrm{Var}(\theta^{\perp}\mid Y)
= \mathbb{E}[\mathrm{Var}(\theta^{\perp}\mid\theta^{\parallel})\mid Y]
+ \mathrm{Var}(\mathbb{E}[\theta^{\perp}\mid\theta^{\parallel}]\mid  Y)\to 0$.
\end{proof}

Finally, we prove the main theorem and corollary in the main text. 

\begin{proof}[Proof of Theorem~\ref{thm:full_consistency}]
The result follows directly from the supplementary lemmas.

Lemma~\ref{lem:consistency_W} establishes posterior consistency of the measurement model:
\(
W^{\parallel}_{dn}\mid Y_{\cdot n} \xrightarrow{p} W^{\parallel,*}_{dn}
\)
as $\lambda\to\infty$.  
Lemma~\ref{lem:consistency_theta_pll} extends this to log-fold changes:
\(
\hat{\theta}^{\parallel}_d \xrightarrow{p} \theta^{\parallel}_d
\)
as $\lambda\to\infty$.  

Lemma~\ref{lem:consistency_mode} then shows that, under the Parzen regularity conditions and sparsity assumption, the Bayes estimator of the scale parameter is consistent:
\(
\bar{\theta}^{\perp} \xrightarrow{p} \theta^{\perp,*}
\)
as $\lambda\to\infty$ and subsequently $D\to\infty$.  

Finally, since the Bayes estimator of the full SSRV is the continuous transformation
\(
\bar{\boldsymbol{\theta}}
= \hat{\theta}^{\parallel} + \bar{\theta}^{\perp}\mathbf{1}_D,
\)
the continuous mapping theorem yields
\begin{equation*}
    \bar{\boldsymbol{\theta}} \xrightarrow{p} \boldsymbol{\theta}^*
\quad \text{as } \lambda\to\infty,\; D\to\infty. 
\end{equation*}
\end{proof}
\begin{proof}[Proof of Corollary~\ref{cor:posterior_consistency_ssrv}]
Theorem~\ref{thm:full_consistency} establishes that the Bayes estimator is consistent:
\(
\mathbb{E}[\theta\mid Y] \xrightarrow{p} \theta^{*}.
\)
Chebyshev’s inequality implies
\[
P\!\left(|\theta_d - \theta^{*}_d| > \varepsilon \,\middle|\, Y\right)
\le \frac{\mathrm{Var}(\theta_d\mid Y)}{\varepsilon^2}.
\]
By the law of total variance and the SSRV conditional independence
\(\theta^{\perp}\perp Y\mid\theta^{\parallel}\),
\[
\mathrm{Var}(\theta_d\mid Y)
=\mathbb{E}[\mathrm{Var}(\theta_d\mid\theta^{\parallel})\mid Y]
+\mathrm{Var}(\mathbb{E}[\theta_d\mid\theta^{\parallel}]\mid Y).
\]
The first term vanishes by assumption,
\(\mathrm{Var}(\theta^{\perp}\mid\theta^{\parallel})\to0\),
and the second vanishes because
\(\mathbb{E}[\theta\mid Y]=\bar{\boldsymbol{\theta}}\xrightarrow{p}\theta^{*}\).
Thus the posterior concentrates:
\[
\theta_d\mid Y \xrightarrow{p}\theta^{*}_d
\]
for every \(d \in \{1,2,...,D\}\).
\end{proof}

\section{Pseudo-code for Scale Model}
\label{sec:pseudo_code}

\begin{algorithm}[H]
\DontPrintSemicolon
\SetAlgoLined
\SetKwInOut{Input}{Input}\SetKwInOut{Output}{Output}
\Input{Y}
\Output{S samples of $\boldsymbol\theta = (\theta_1,...,\theta_D)$}
\BlankLine
\For{s in \{1,...,S\}}{
  \vspace{0.5em}
  Sample $\boldsymbol{W^{\parallel(s)}_{.n}} \sim p(\boldsymbol{Y_{.n}})$ \vspace{0.4em} \;
  Calculate $ \boldsymbol{\theta^{\parallel(s)}} = f(\boldsymbol{W^{\parallel(s)}_{.n}})$ \;
  \BlankLine
  Find $\theta^{\perp(s)} = \underset{t \in [t_l, t_u]}{-\text{argmax }} p_D(t;\theta^{\parallel(s)}_1, \theta^{\parallel(s)}_2, ...,\theta^{\parallel(s)}_D)$ \;
  Sample $\boldsymbol{\epsilon} \sim N(0,\tau^2)$ \;
  Calculate $\boldsymbol{\theta^{(s)}}$ = $\boldsymbol{\theta^{\parallel(s)}} + \theta^{\perp(s)}\boldsymbol{1_D} + \boldsymbol{\epsilon}$ 
}
\caption{Sparse SSRV Incorporating Uncertainty through Laplace Approximation}
\end{algorithm}


\begin{algorithm}[H]
\DontPrintSemicolon
\SetAlgoLined
\SetKwInOut{Input}{Input}\SetKwInOut{Output}{Output}
\Input{Y}
\Output{S samples of $\theta$}
\BlankLine
\For{s in \{1,...,S\}}{
  \vspace{0.5em}
  Sample $\boldsymbol{W^{\parallel(s)}_{.n}} \sim p(\boldsymbol{Y_{.n}}) $ \;
  Bootstrap over N : $W^{\parallel(s)}_{d., \hspace{0.1em} boot}\sim W^{\parallel(s)}_{d.}$ \;
  Calculate $\boldsymbol{\theta^{\parallel(s)}} = f(\boldsymbol{W^{\parallel(s)}_{boot}})$ \;
  Bootstrap over D : $\boldsymbol{\theta^{\parallel(s)}_{boot}} \sim \boldsymbol{\theta^{\parallel(s)}}$ \;
  Find $\theta^{\perp(s)} = \underset{t \in [t_l, t_u]}{-\text{argmax }} p_D(t;\theta^{\parallel(s)}_{1,boot}, \theta^{\parallel(s)}_{2,boot}, ...,\theta^{\parallel(s)}_{D,boot})$ \;
  Calculate $\boldsymbol{\theta^{(s)}}$ = $\boldsymbol{\theta^{\parallel(s)}} + \theta^{\perp(s)}\boldsymbol{1_D}$ 
}
\caption{Sparse SSRV Incorporating Uncertainty through Bootstrapping}
\end{algorithm}

\section{Details for Simulation Studies}
\label{sec:sim-benchmark-details}
To mimic realistic scenarios, where increases and decreases in abundance are asymmetric, we simulated data with 80\% of the relevant features having a positive effect and 20\% a negative effect. The effect sizes were independently drawn from an exponential-like distribution ranging from 1 to 6, favoring smaller values. Each scenario was independently replicated 100 times, and results are presented as averages over these replicates. The results are reported in the main article (Figure~\ref{fig:simulation}). We summarized the performance of each method using FDR and TPR, and additionally reported the \(F_{0.5}\) score (\(F_{0.5} = \frac{(1+0.5^2)\cdot\operatorname{TP}}{(1+0.5^2)\cdot\operatorname{TP}+0.5^2\cdot\operatorname{FN}+\operatorname{FP}}\)), which combines these two measures while placing a slightly greater emphasis on FDR. 

As an alternative scenario to Simulation Studies in Section~\ref{sec:benchmarks-simulation}, we repeated the simulation study but with different effect size ratios between positive and negative effect sizes. Here, among truly differentially abundant features, 80\% of them are decreasing in its abundance while 20\% of them are increasing. The results in Figure~\ref{fig:sim_supp} exhibit patterns consistent with those observed in the main simulation study.

\begin{figure}[ht!]
    \centering
    \includegraphics[width=1\linewidth]{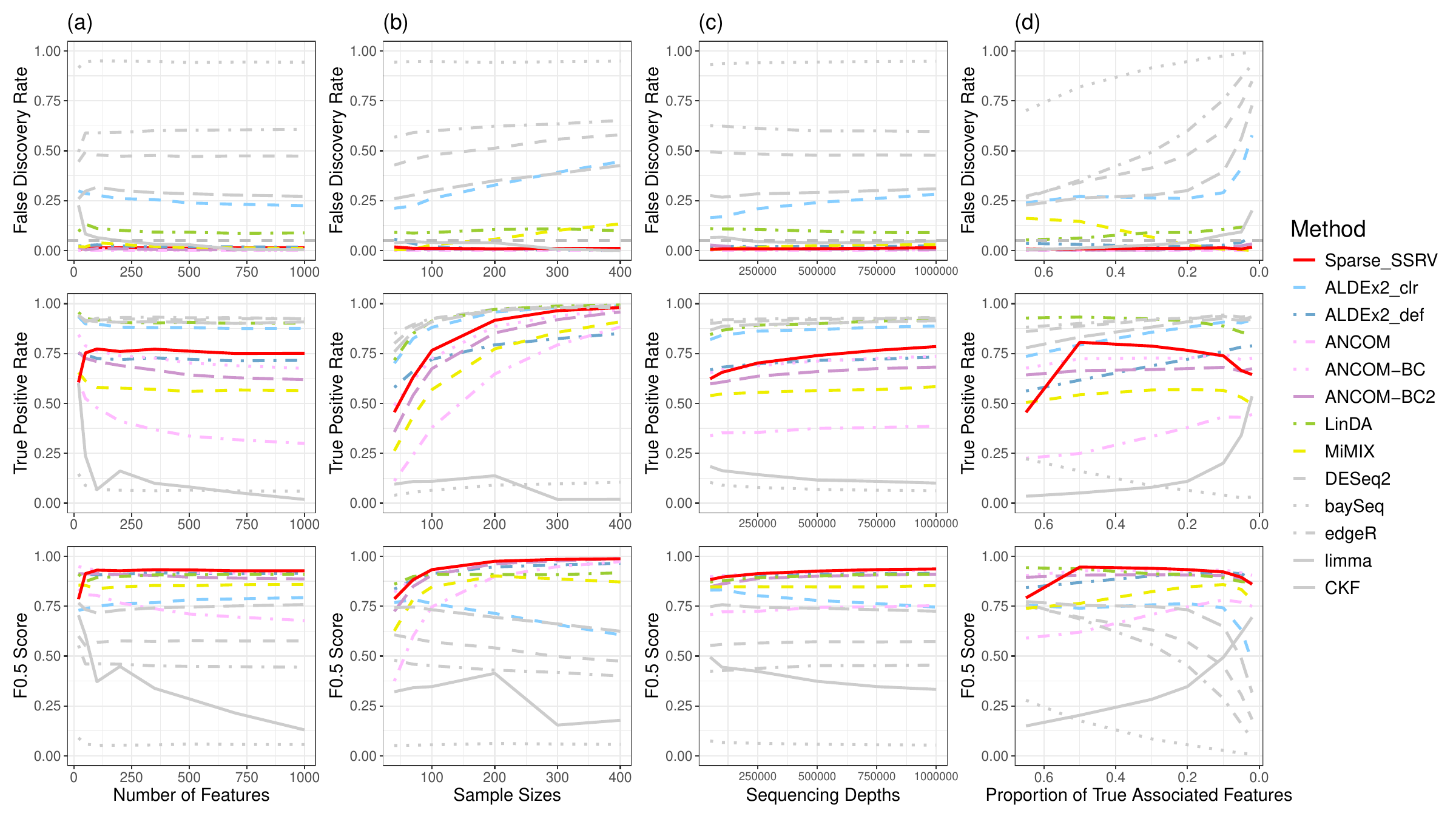}
    \caption{Simulation Study for an alternative scenario where larger number of features are increasing than features that are decreasing. Other details for the setting remain the same as the main simulation studies.}
    \label{fig:sim_supp}
\end{figure}



\section{Details for Real Data Studies} \label{sec:detail_real}
\subsection{McMurrough et al., 2014}
We compared microbial abundance of 1,600 sequence variants for 14 samples, with 7 selected and 7 non-selected conditions. (McMurrough et al., 2014) identified ground truths through validation experiments. We used preprocessed data from ALDEx2 Bioconductor package.
\begin{figure}[ht]
    \centering
    \includegraphics[width=1\linewidth]{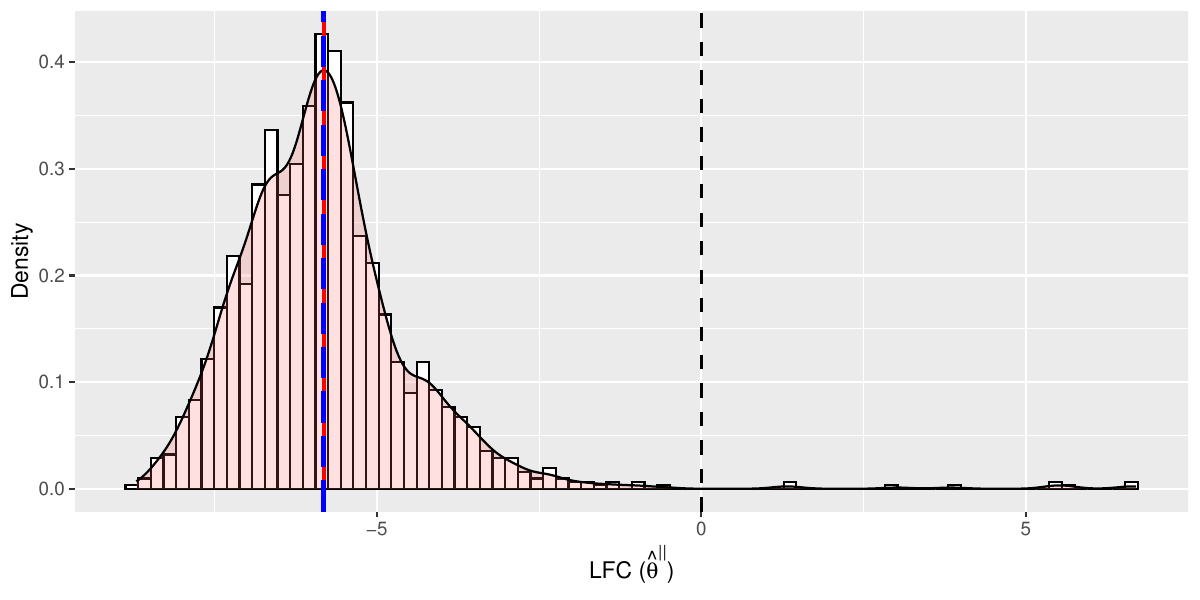}
    \caption{Distribution of LFC values computed from observed reads ($\hat\theta^{\parallel}$) in McMurrough et al. Red line is a mode found using only null features, which we consider as the true mode of null features. Blue dash line represent mode estimated based on LFC values of all features. If red and blue dash lines overlap, it means the underlying system is sparse since the mode found using all features represents the location of LFC values of null features accurately. For this data, we can see red line and blue line overlaps completely, indicating the underlying system is sparse.}
    \label{fig:lfc_mcmurrough}
\end{figure}

\subsection{Zemb et al., 2020}
We compared microbial abundance of taxa in species level for 10 fecal samples where 5 are added with E.coli and 5 are control. These two group of samples are otherwise stable, implying E.coli is the only differentially abundant feature between conditions. The processed data was shared by the author of \cite{zemb2020absolute}. The data initially included 428 species. We then filtered out taxa with an average read count below five, following the same procedure used in our simulation studies.

\begin{figure}[ht]
    \centering
    \includegraphics[width=1\linewidth]{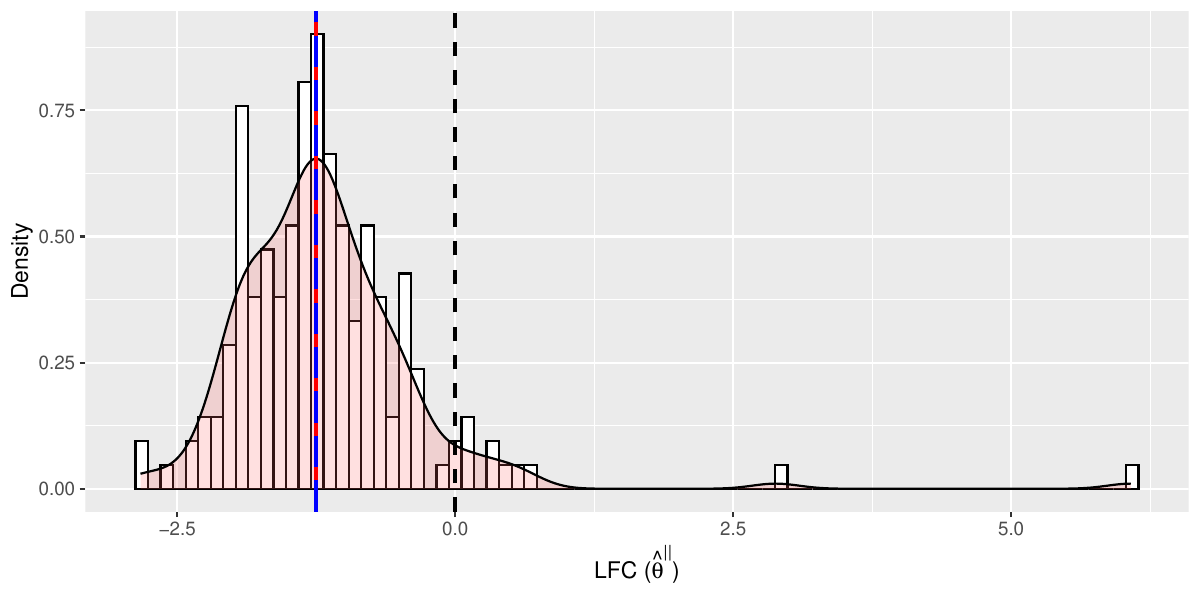}
    \caption{Distribution of LFC values computed from observed reads ($\hat\theta^{\parallel}$) in Zemb et al. Again, red and blue lines overlap completely, suggesting the underlying system is sparse.}
    \label{fig:lfc_zemb}
\end{figure}

\subsection{Stämmler et al., 2016}
In this study, pooled fecal samples of mice are divided to 6 pools and then spiked with Salinibacter ruber, Rhizobium radiobacter and Alicyclobacillus acidiphilus. The amount of spike-in of R.radiobacter and A.acidiphilus are different across pools, whereas the same amount of Salinibacter ruber is added to all pools. Hence, A.acidiphilus and R.radiobacter are the only taxa that are differentially abundant between any pair of pools. The concentration of A. acidiphilus increases progressively from Pool 1 ($1.00\times10^7$) to Pool 6 ($2.43 \times 10^9$). In contrast, R. radiobacter has the highest concentration in Pool 1 ($2.43 \times 10^9$), which gradually decreases through Pool 2 to Pool 6 ($1.00 \times 10^7$). As a result, comparison of Pool 1 and Pool 6 is the perfectly symmetric case---LFC of A.acidiphilus is the exact negative of the one of R.radiobacter. This is also evident from Figure~\ref{fig:lfc_stammler} as two LFC values that are symmetrically located with respect to the main cluster of LFC values.

\begin{figure}[ht]
    \centering
    \includegraphics[width=1\linewidth]{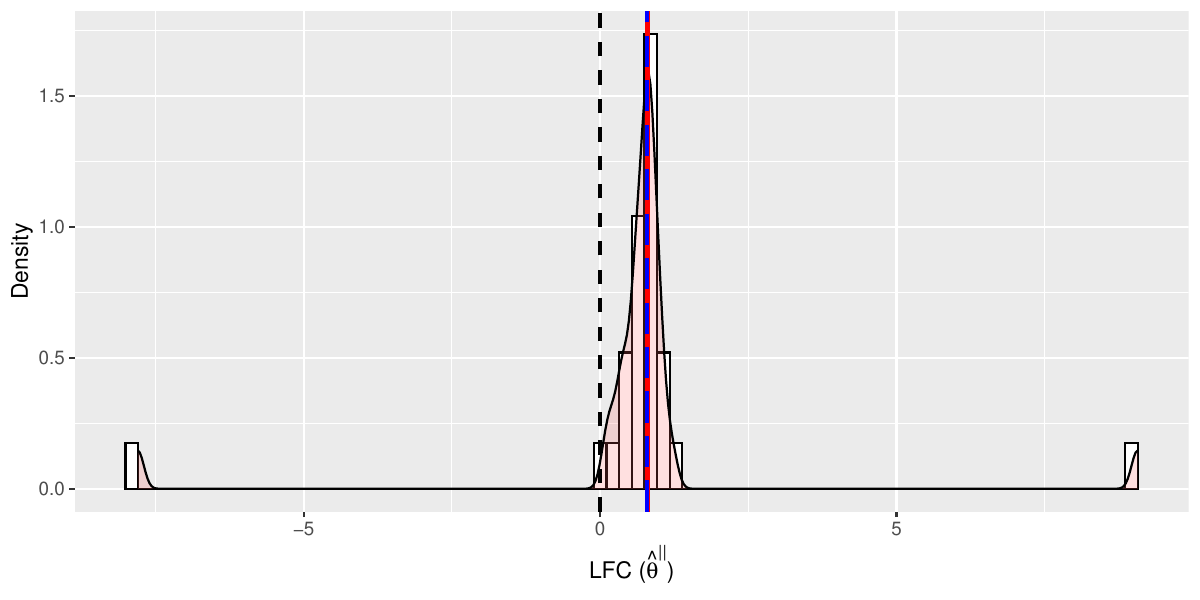}
    \caption{Distribution of LFC values computed from observed reads ($\hat\theta^{\parallel}$) in Stämmler et al. Red and blue lines overlap completely, implying the underlying system is sparse.}
    \label{fig:lfc_stammler}
\end{figure}

\subsection{Jin et al., 2022} \label{sec:supp_jin}
Jin et al. collected sample from mice at distal and proximal location of a cecum and measured cell-based operational taxonomic unit (cOTU). The absolute cell abundance is quantified using BarBIQ, the method proposed in \cite{jin2022quantsinglecell}.  We compared the abundance of microbes at the genus level across 16 samples, with 8 collected from mice on a vitamin A–sufficient diet and the other 8 from mice on a vitamin A–deficient diet. Since there is no known ground truth in this study, we used results from \textit{Gold standard model} to identify taxa to serve as a gold reference, adopted from \cite{nixon2024a}. \textit{Gold standard model} uses ALDEx2 with scale model informed by microbial load and uncertainty in the estimates. Scale model within ALDEx2 is defined as following
$$\log W^{\perp}_n \sim N(\log\mu_n, \gamma^2)$$
where $\mu_n$ represents the microbial load of sample n and $\gamma^2$ represent measurement noise. We found the microbial load and the standard deviation of microbial load measurement in Supplementary Data 8 provided in Jin et al.

\begin{figure}[ht!]
    \centering
    \includegraphics[width=1\linewidth]{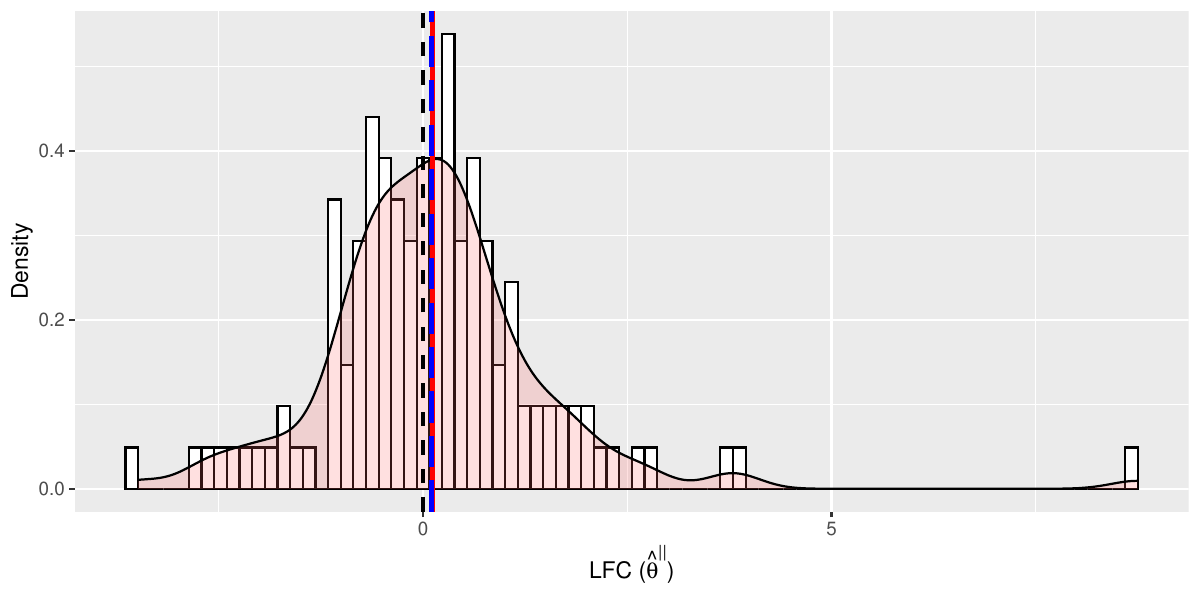}
    \caption{Distribution of LFC values computed from observed reads ($\hat\theta^{\parallel}$) in Jin et al. Red and blue lines overlap completely, implying the underlying system is sparse.}
    \label{fig:lfc_jin}
\end{figure}

\subsection{Prasse et al., 2015}
We compared the abundance of taxa at the genus level across 57 samples from natural and restored tidal freshwater wetlands, with 24 samples from natural sites and 33 from restored sites. To get ground truth, we used \textit{Gold standard model}. Since no standard deviation was provided for the microbial load measurements and no replicates were available, we chose uncertainty $\gamma^2 = 0.5$ which is default value for scale model in ALDEx2. Sensitivity analysis results in Figure~\ref{fig:prasse_sensitivity} shows results based on different $\gamma$ values. As $\gamma$ increases, the Gold Standard Model returns a smaller number of significant features. Although this affects the number of features labeled as false positives and true positives across different methods, the overall performance remains consistent for values of $\gamma$ up to 1.5. The data was preprocessed and uploaded on Github by the authors of \cite{eppschmidt2023benchmarkdata}. We additionally filter out taxa with prevalence less than 10\%. Here, we filtered taxa based on prevalence rather than mean read counts, as mean-based filtering did not effectively remove spurious signals, due to the particularly high variance in taxon abundances across samples in this dataset. We chose prevalence-based filtering procedure as it is already embedded in some of the competing methods including LinDA, ANCOM-BC2, and edgeR.

\begin{figure}[ht!]
    \centering
    \includegraphics[width=1\linewidth]{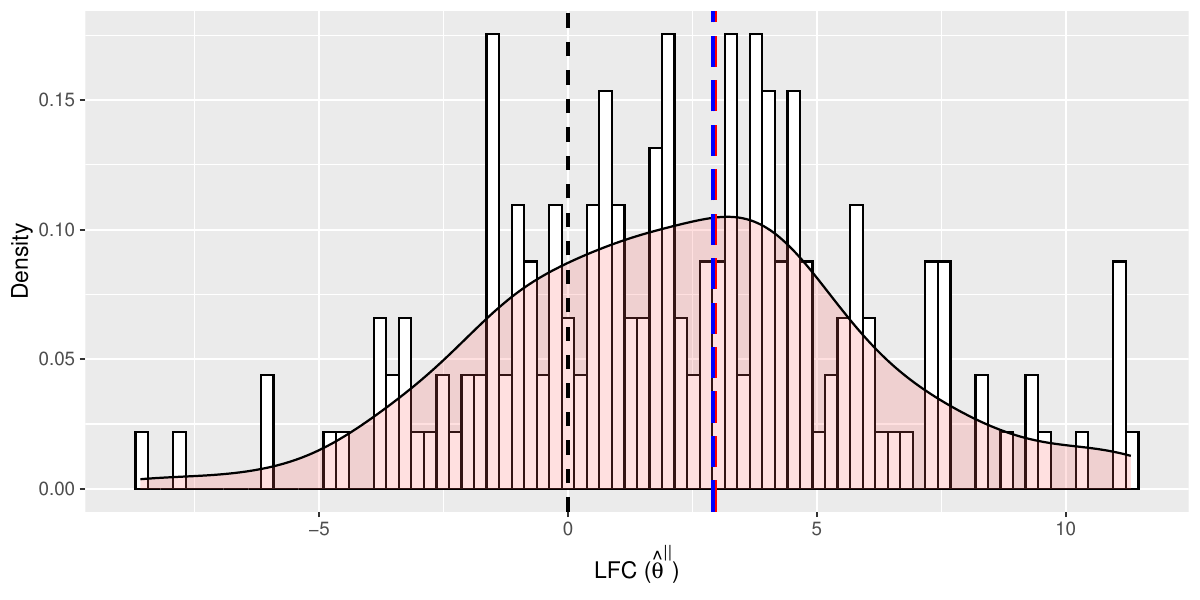}
    \caption{Distribution of LFC values computed from observed reads ($\hat\theta^{\parallel}$) in Prasse et al. study. Here, the red and blue lines almost overlap, but the distribution is so spread out that the density at the blue line and other locations do not differ by a lot. This suggests that the mode obtained based on all LFC values (blue line) is placed within LFCs of null features, but the mode is not very distinct. This implies the underlying system is sparse but with presence of large measurement noise.}
    \label{fig:lfc_prasse}
\end{figure}

\begin{figure}[ht]
    \centering
    \includegraphics[width=1\linewidth]{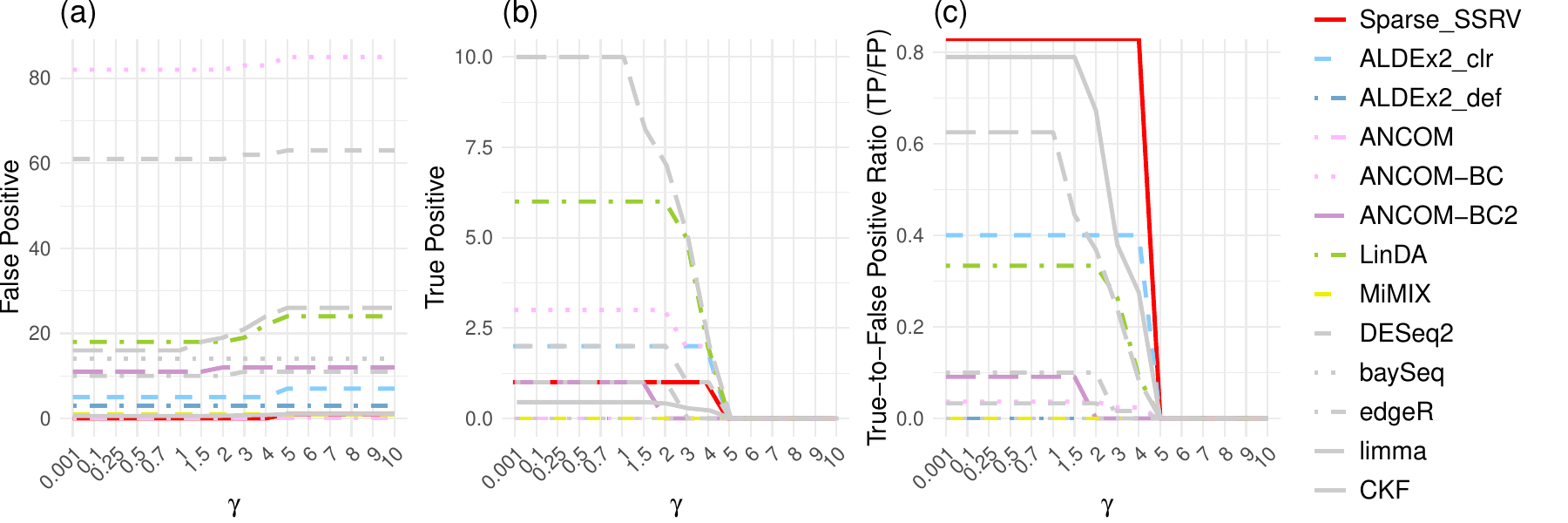}
    \caption{False positive, true positive, and true-to-false positive ratio over varying scale uncertainty ($\gamma$) specified in Gold standard model. Method-specific colors and line types in this figure match those used in Figure~\ref{fig:simulation} for consistency.}
    \label{fig:prasse_sensitivity}
\end{figure}

\subsection{Vandeputte et al., 2017}
We compared abundance of taxa in genus level for 95 samples where 29 are patients with Crohn's Disease and 66 are healthy controls. 
To get ground truth, we used results from \textit{Gold standard model} as in \cite{nixon2024a}. For the scale model 
$$\log W^{\perp}_n \sim N(\log\mu_n, \gamma^2),$$
microbial load of sample n ($\mu_n$) is measured by flow-cytometry. $\gamma^2$ is set to $0.7$, based on the standard deviation of technical replicates of 40 samples. Details on the preprocessing procedure and sensitivity analysis results for choice of $\gamma^2$ can be found in \cite{nixon2024a}.

\begin{figure}[ht!]
    \centering
    \includegraphics[width=1\linewidth]{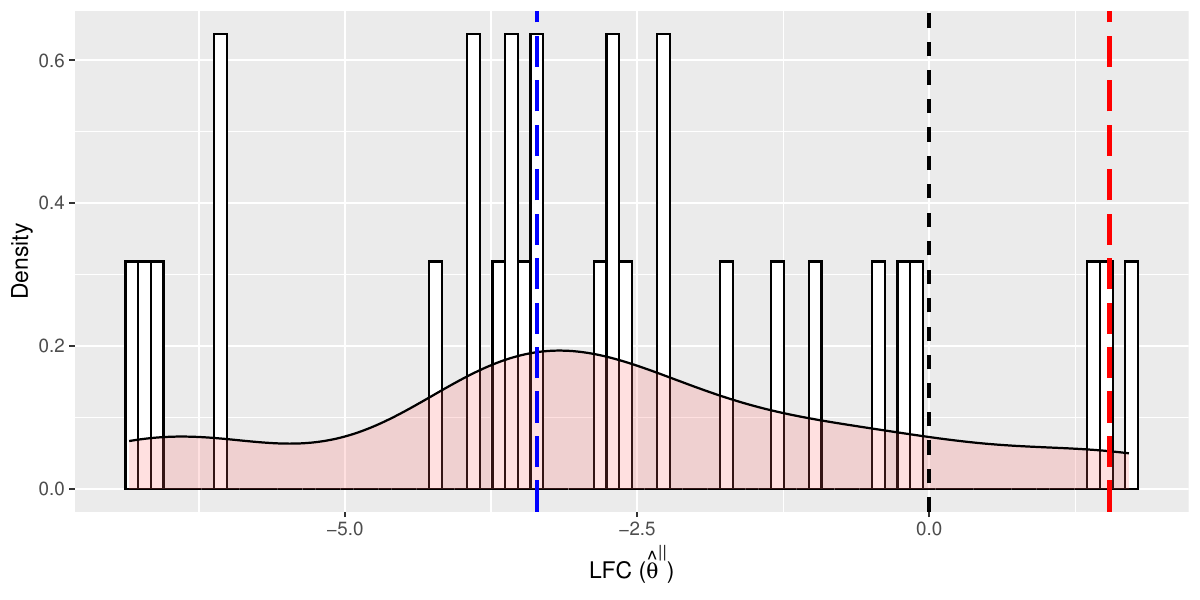}
    \caption{Distribution of LFC values computed from observed reads ($\hat\theta^{\parallel}$) in Stämmler et al. In this case, the red and blue lines are clearly separated, suggesting that the mode obtained based on all LFC values (blue line) do not capture null features. This indicates violation of sparse system assumption.}
    \label{fig:lfc_vandeputte}
\end{figure}

\end{document}